\documentclass[final]{siamltex}
\usepackage{bm} % for bold math symbols, \bm command
\usepackage{amsfonts} % for real number set symbol, \mathcal{R}:
\usepackage{multirow} % for multirow tables
\usepackage{relsize}  % for \mathlarger command
\usepackage{adjustbox}
\usepackage[dvips]{geometry}
\usepackage{color} %include even if images aren’t in color
\usepackage{epsfig}
\usepackage{latexsym}
\usepackage{pstricks}
\usepackage{url}
\newcommand{\tr}{^{\sf T}}      % transpose
\newcommand{\m}[1]{{\bf{#1}}}   % bold vectors
\newcommand{\g}[1]{\bm #1}      % bold greek letters
\newcommand{\C}[1]{{\cal {#1}}} % calligraphic font
\newcommand{\underscore}{\underline{\hspace{0.25cm}}}

\newcommand{\AI}{\m{H}}

\newtheorem{remark}{Remark}[section]

\title{
A Multilevel Bilinear Programming Algorithm For the Vertex Separator Problem
\thanks{
May 25, 2016.
The research was supported by
the Office of Naval Research under Grants N00014-11-1-0068
and N00014-15-1-2048 and
by the National Science Foundation under grants 1522629
and 1522751.
Part of the research was performed while the second author
was a Givens Associate at Argonne National Laboratory.
}
}

\author{
        William W. Hager\thanks{{\tt hager@ufl.edu},
        http://people.clas.ufl.edu/hager/
        PO Box 118105,
        Department of Mathematics,
        University of Florida, Gainesville, FL 32611-8105.
        Phone (352) 294-2308. Fax (352) 392-8357.}
\and
        James T. Hungerford\thanks{{\tt jamesthungerford@gmail.com},
        M.A.I.O.R. -- Management, Artificial Intelligence, and Operations
        Research, Srl., Lucca, ITALY}
\and
        Ilya Safro\thanks{{\tt isafro@clemson.edu},
        http://www.cs.clemson.edu/$\sim$isafro,
        228 McAdams Hall,
        School of Computing,
        Clemson University, Clemson, SC 29634.
        Phone (864) 656-0637.}
}
\begin{document}

\maketitle

\begin{abstract}
The Vertex Separator Problem for a graph is
to find the smallest collection of
vertices whose removal breaks the graph into two disconnected subsets that
satisfy specified size constraints.
The Vertex Separator Problem was formulated in the paper
10.1016/j.ejor.2014.05.042
as a continuous (non-concave/non-convex) bilinear quadratic program.
In this paper, we develop a more general continuous bilinear program
which incorporates vertex weights, and which applies to the coarse graphs
that are generated in a multilevel compression of the original
Vertex Separator Problem.
A Mountain Climbing Algorithm is used to find a stationary point of
the bilinear program, while 
perturbation techniques are used to either dislodge an iterate
from a saddle point or escape from a local optimum.
We determine the smallest possible perturbation that will
force the current iterate to a different location, with a
possibly better separator.
The algorithms for solving the bilinear program are employed
during the solution and refinement phases in a multilevel scheme.
Computational results and comparisons demonstrate the
advantage of the proposed algorithm.
\end{abstract}

% to remove -----------


\pagestyle{myheadings} \thispagestyle{plain}
\markboth{W. W. HAGER, J. T. HUNGERFORD, AND I. SAFRO}
{A MULTILEVEL ALGORITHM FOR THE VERTEX SEPARATOR PROBLEM}
% -------

%-------------------------------------------------------------------------------
\section{Introduction}
%-------------------------------------------------------------------------------
Let $G = (\C{V}, \C{E})$ be a graph on vertex set $\C{V} = \{1,2,\ldots,n\}$ and
edge set $\C{E} \subseteq \C{V} \times \C{V}$.
We assume $G$ is simple and
undirected; that is for any vertices $i$ and $j$ we have
$(i,i) \notin \C{E}$ and
$(i,j) \in \C{E}$ if and only if $(j,i) \in \C{E}$ (note that this implies that
$|\C{E}|$, the number of elements in $\C{E}$,
is twice the total number of edges in $G$).
For each $i \in \C{V}$, let
$c_i\in\mathbb{R}$
denote the cost and $w_i > 0$ denote the weight of vertex $i$.
If $\C{Z} \subseteq \C{V}$, then
\[
\C{C}(\C{Z}) = \sum_{i \in \C{Z}} c_i \quad \mbox{and} \quad
\C{W}(\C{Z}) = \sum_{i \in \C{Z}} w_i
\]
denote the total cost and weight of the vertices in $\C{Z}$, respectively.

If the vertices $\C{V}$ are partitioned into three disjoint sets
$\C{A}$, $\C{B}$, and $\C{S}$, then $\C{S}$ \emph{separates} $\C{A}$ and 
$\C{B}$ if there is no edge $(i,j)\in\C{E}$ with $i\in\C{A}$ and $j\in\C{B}$.
The Vertex Separator Problem (VSP) is to minimize the cost 
of $\C{S}$ 
while requiring that $\C{A}$ and $\C{B}$ have
approximately the same weight.
%Since minimizing $\C{C}(\C{S})$ is equivalent to maximizing 
%$\C{C}(\C{A} \cup \C{B})$, 
We formally state the VSP as follows:
\begin{eqnarray}\label{VSP}
% objective
& \displaystyle {\min_{\C{A},\C{S},\C{B} \subseteq \C{V}} \quad \C{C}(\C{S})} & \\
% constraints
& \displaystyle {\mbox{subject to} \quad
\C{S} = \C{V} \setminus (\C{A}\cup\C{B}), \quad
\C{A} \cap \C{B} = \emptyset, \quad
(\C{A} \times \C{B}) \cap \C{E} = \emptyset, } & \nonumber \\
& \displaystyle {\ell_a \le \C{W}(\C{A}) \le u_a, \;\; \mbox{and} \;\;
  \ell_b \le \C{W}(\C{B}) \le u_b}, & \nonumber
\end{eqnarray}
where $\ell_a$, $\ell_b$, $u_a$, and $u_b$ are given nonnegative
real numbers less than or equal to $\C{W}(\C{V})$.
The constraints
$\C{S} = \C{V} \setminus (\C{A}\cup\C{B})$
and $\C{A} \cap \C{B} = \emptyset$ ensure that $\C{V}$ is partitioned into
disjoint sets $\C{A}$, $\C{B}$, and $\C{S}$, while the constraint
$(\C{A} \times \C{B}) \cap \C{E} = \emptyset$ ensures that there are
no edges between the sets $\C{A}$ and $\C{B}$.
Throughout the paper, we assume (\ref{VSP}) is feasible.
In particular, if $\ell_a, \ell_b \ge 1$, then there exist at least two
distinct vertices $i$ and $j$ such that $(i,j)\notin \C{E}$; that is, $G$ is not
a complete graph.

Vertex separators have applications in VLSI design \cite{KernighanLin70,
Leiserson80, Ullman84}, finite element methods
\cite{MillerTengThurstonVavasis98}, parallel processing \cite{Evrendilek2008},
sparse matrix factorizations
(\cite[Sect.~7.6]{Davis06book},
\cite[Chapter~8]{GeorgeLiu}, and
\cite{PothenSimonLiou90}), hypergraph
partitioning \cite{KayaaslanPinar},
and network security \cite{raey, kim2013combinatoric,LitsasPPS13}.
The VSP is NP-hard \cite{BuiJones92, Fukuyama2006}. However, due to its
practical significance, many heuristics have been developed for obtaining
approximate solutions, including node-swapping heuristics
\cite{LeisersonLewis89}, spectral methods \cite{PothenSimonLiou90}, semidefinite
programming methods \cite{FeigeHajiaghayiLee2008}, and recently a breakout local
search algorithm \cite{BenlicHao}.

It has been demonstrated repeatedly that for problems on large-scale graphs,
such as finding minimum $k$-partitionings
\cite{bmsss13,Hendrickson,KarypisKumar98e}
or minimum linear arrangements \cite{ron2011relaxation,SafroRonBrandt06a},
optimization algorithms can be much more effective when carried out
in a multilevel framework.
In a multilevel framework, a hierarchy of increasingly smaller graphs
is generated which approximate the original graph,
but with fewer degrees of freedom.
The problem is solved for the coarsest graph in the hierarchy,
and the solution is gradually uncoarsened and refined to obtain
a solution for the original graph.
During the uncoarsening phase, optimization algorithms are commonly
employed locally  to make fast  improvements to the solution at each
level in the algorithm.
Although multilevel algorithms are inexact
for most NP-hard problems on graphs, they typically produce very high
quality solutions and
are very fast (often linear in the number of vertices plus the number
of edges with no hidden coefficients).

Early methods \cite{GilbertZmijewski87, PothenSimonLiou90},
for computing vertex separators were based on computing
edge separators (bipartitions of $\C{V}$ with low cost edge-cuts).
In these
algorithms, vertex separators are obtained from edge separators by selecting
vertices incident to the edges in the cut. More recently, \cite{AcerKayaaslan}
gave a method for computing vertex separators in a graph
by finding low cost net-cuts in an associated hypergraph.
Some of the most widely used heuristics for computing edge separators are the
node swapping heuristics of Fiduccia-Mattheyses
\cite{FiducciaMattheyses82} and Kernighan-Lin \cite{KernighanLin70}, in 
which vertices are exchanged between sets until the current partition is
considered to be locally optimal.
Many multilevel edge separator
algorithms have been developed and incorporated into
graph partitioning packages (see survey in \cite{bmsss13}).
% such as METIS \cite{KarypisKumar95,KarypisKumar98e},
% JOSTLE \cite{JOSTLE}, and ZOLTAN \cite{ZOLTAN}.
In \cite{AshcraftLiu94},
a Fiduccia-Mattheyses type heuristic is used to find vertex separators
directly. Variants of this algorithm have been incorporated into the multilevel
graph partitioners METIS \cite{KarypisKumar95, KarypisKumar98e} and BEND
\cite{HendricksonRothberg97b}.

In \cite{HagerHungerford14}, the authors make a departure from traditional
discrete-based heuristics for solving the VSP, and present the first formulation
of the problem as a continuous optimization problem. In particular, 
when the vertex weights are identically one,
conditions
are given under which the VSP is equivalent to solving a continuous bilinear
quadratic program.

The preliminary numerical results of \cite{HagerHungerford14} indicate that
the bilinear programming formulation can serve as an effective tool
for making local improvements to a solution in a multilevel context.
The current work makes the following contributions:
\smallskip
\begin{itemize}
\item [1.] The bilinear programming model of \cite{HagerHungerford14} is
extended to the case where vertex weights are possibly greater than one.
This generalization is important since each vertex in a multilevel
compression corresponds to a collection of vertices in the original graph.
The bilinear formulation of the compressed graph is not exactly
equivalent to the VSP for the compressed graph, but it very closely
approximates the VSP as we show.
%\item [2.] Conditions are given under which local optimality of a partition
%(in the Fiduccia-Mattheyses sense) implies local optimality in the bilinear
%program. A weak converse is also shown to hold.
\item[2.] Optimization algorithms applied to the bilinear VSP model
converge to stationary points which may not be global optima.
Two techniques are developed to escape from a stationary point
and explore a new part of the solution space.
One technique uses the first-order optimality conditions to
construct an infinitesimal perturbation of the objective function
with the property that a gradient descent step
moves the iterate to a new location
where the separator could be smaller.
This technique is particularly effective when the current iterate lies
at a saddle point rather than a local optimum.
The second technique involves relaxing the constraint
that there are no edges between the sets in the partition.
Since this constraint is enforced by a penalty in the objective,
we determine the smallest possible relaxation of the penalty for
which a gradient descent step moves the iterate to a new location where the
separator could be smaller.
\item[3.] A multilevel algorithm is developed which incorporates the weighted
bilinear program in the refinement phase along with the techniques to
escape from a stationary point.
Computational results are given to compare the quality of the solutions
obtained with the bilinear programming approach to a multilevel vertex
separator routine in the METIS package.
The algorithm is shown to be especially effective on p2p networks and graphs
having heavy-tailed degree distributions.
\end{itemize}
\smallskip

The outline of the paper is as follows. Section \ref{sectProblem} reviews the
bilinear programming formulation of the VSP in \cite{HagerHungerford14}
and develops the weighted formulation which is suitable for the coarser
levels in our algorithm.
Section \ref{sectSolve} presents an algorithm for finding
approximate solutions to the bilinear program and develops two techniques
for escaping from a stationary point.
Section \ref{sectAlgorithm} summarizes the multilevel framework, while
Section \ref{sectResults} gives numerical results
comparing our algorithm to METIS.
Conclusions are drawn in Section \ref{sectConclude}.

\section{Notation}
Vectors or matrices whose entries are all $0$ or all $1$
are denoted by $\m{0}$ or $\m{1}$ respectively,
where the dimension will be clear from the context.
If $\m{x} \in \mathbb{R}^n$ and
$f:\mathbb{R}^n\rightarrow \mathbb{R}$, then $\nabla f (\m{x})$ denotes
the gradient of $f$ at $\m{x}$, a row vector, and $\nabla^2 f (\m{x})$
is the Hessian. If $f : \mathbb{R}^n \times \mathbb{R}^n \rightarrow
\mathbb{R}$,
then $\nabla_{\m{x}} f (\m{x},\m{y})$ is the row vector corresponding to the
first $n$ entries of $\nabla f (\m{x},\m{y})$,
while $\nabla_{\m{y}} f (\m{x},\m{y})$
is the row vector
corresponding to the last $n$ entries.
If $\m{A}$ is a matrix, then $\m{A}_i$ denotes the $i$-th row of $\m{A}$.
If $\m{x}\in\mathbb{R}^n$,
then $\m{x}\ge \m{0}$ means $x_i \ge 0$ for all $i$,
and $\m{x}\tr$ denotes the transpose, a row vector.
Let $\m{I}\in\mathbb{R}^{n\times n}$ denote the $n\times n$ identity matrix,
let $\m{e}_i$ denote the $i$-th column of $\m{I}$,
and let $|\C{A}|$ denote the number of elements in the set $\C{A}$.
%If $\m{X}$ is a set and $f$ is a real-valued function defined over $\m{X}$,
%then a \emph{feasible point} of the
%optimization problem $\max\;\{f (\m{x}) : \m{x} \in \m{X} \}$ is a point
%$\m{x} \in \m{X}$, and an \emph{optimal solution}
%is a feasible point $\m{x}^*$ such
%that $f (\m{x}^*) \ge f (\m{x})$ for every $\m{x} \in \m{X}$.
%A (\emph{strict}) \emph{local maximizer}
%is a point $\hat{\m{x}} \in \m{X}$ such that for some
%$\epsilon > 0$, we have $f (\hat{\m{x}}) \ge f (\m{x})$
%($f (\hat{\m{x}}) > f (\m{x})$) for every
%$\m{x} \in \m{X}$ such that $||\m{x} - \hat{\m{x}}||_2 \le \epsilon$.
%If $\C{Z} \subseteq \C{V}$, then the closed neighborhood of $\C{Z}$ is
%the set defined by
%
%\[
%\C{N}[\C{Z}] = \C{Z} \cup \{j \in \C{V}\setminus \C{Z} : (i,j) \in \C{E}
%\mbox{ for some } i \in \C{Z}\}.
%\]
%

%------------------------------------------------------------------------------%
\section{Bilinear programming formulation}
\label{sectProblem}
%------------------------------------------------------------------------------%
Since minimizing $\C{C}(\C{S})$ in (\ref{VSP}) is equivalent to maximizing 
$\C{C}(\C{A} \cup \C{B})$, we may view the VSP as the following 
maximization problem:
\begin{eqnarray}\label{VSP2}
% objective
& \displaystyle {\max_{\C{A},\C{B} \subseteq \C{V}} \quad \C{C}(\C{A}\cup\C{B})} 
& \\
% constraints
& \displaystyle {\mbox{subject to} \quad
\C{A} \cap \C{B} = \emptyset, \quad
(\C{A} \times \C{B}) \cap \C{E} = \emptyset, } & \nonumber \\
& \displaystyle {\ell_a \le \C{W}(\C{A}) \le u_a, \;\; \mbox{and} \;\;
\ell_b \le \C{W}(\C{B}) \le u_b}.  & \nonumber
\end{eqnarray}
Next, 
observe that any pair of subsets $\C{A}, \C{B} \subseteq \C{V}$ is associated with
a pair of incidence vectors $\m{x},\m{y} \in \{0,1\}^{n}$ defined by
\begin{equation}\label{xydef}
x_i = \left\{
\begin{array}{rrr}
1, & \mbox{if} \; i \in \C{A}\\
0, & \mbox{if} \; i \notin \C{A}
\end{array}
\right.
\quad \mbox{and} \quad
y_i = \left\{
\begin{array}{rrr}
1, & \mbox{if} \; i \in \C{B}\\
0, & \mbox{if} \; i \notin \C{B}
\end{array}
\right. .
\end{equation}
Let $\AI := (\m{A} + \m{I})$, where
$\m{A}$ is the $n \times n$ 
adjacency matrix for $G$ defined by $a_{ij} = 1$ if $(i,j) \in \C{E}$ and 
$a_{ij} = 0$ otherwise, and $\m{I}$ is the $n\times n$ identity matrix. Note
that since $G$ is undirected, $\AI$ is symmetric.
Then we have
\begin{eqnarray}
{\m{x}\tr\AI\m{y}}\quad & {=} &\quad
{\sum_{i = 1}^n\sum_{j = 1}^n x_i a_{ij} y_j}
\; \, {+} \;
{\sum_{i = 1}^n x_i y_i} \nonumber\\
& {=} & \quad
{\sum_{x_i = 1}\sum_{y_j = 1} a_{ij}}
\; \; \quad {+} \; {\sum_{x_i = y_i= 1} 1} \nonumber \\
& {=} & \quad
{\sum_{i \in\C{A}}\sum_{j\in\C{B}} a_{ij}} \; \quad \quad {+} \;
{\sum_{i \in \C{A} \cap \C{B}} 1} \nonumber \\
& {=} & \quad
{|(\C{A}\times\C{B})\cap\C{E}|} \, \; \;\; {+} \; \;
{|\C{A} \cap \C{B}|}.  \label{penaltymeaning}
\end{eqnarray}
So, the constraints
$\C{A} \cap \C{B} = \emptyset$ and
$(\C{A} \times \C{B}) \cap \C{E} = \emptyset$
in (\ref{VSP2}) hold if and only if 
\begin{equation}\label{sepcond}
\m{x}\tr\AI\m{y} = 0.
\end{equation}
Hence, a binary formulation of (\ref{VSP2}) is
\begin{eqnarray}
% objective
& \displaystyle {\max_{\m{x}, \m{y} \in \{0, 1\}^n} \;\;
\m{c}\tr(\m{x} + \m{y})} & \nonumber \\
% constraints
\displaystyle {\mbox{subject to}} & 
\m{x}\tr\AI\m{y} = 0, \label{preBVSP}\\
& \ell_a \le \m{w}\tr\m{x}  \le u_a, \;\; \mbox{and} \;\;
  \ell_b \le \m{w}\tr\m{y}  \le u_b, \nonumber
\nonumber
\end{eqnarray}
where $\m{c}$ and $\m{w}$ are the $n-$dimensional vectors which
store the costs $c_i$ and weights $w_i$ of vertices, respectively.
 
Now, consider the following problem in which the quadratic constraint
of (\ref{preBVSP}) has been relaxed:
\begin{eqnarray}\label{BVSP}
% objective
& \displaystyle {\max_{\m{x}, \m{y} \in \{0, 1\}^n} \;\;
f (\m{x}, \m{y}) := \; \m{c}\tr(\m{x} + \m{y}) -
\gamma\m{x}\tr\AI\m{y}} & \\
% constraints
& \begin{array}{cl} \mbox{subject to} &
\ell_a \le \m{w}\tr\m{x}  \le u_a\;\; \mbox{ and }\;\;
\ell_b \le \m{w}\tr\m{y}  \le u_b.
\end{array}
\nonumber
\end{eqnarray}
Here, $\gamma \in \mathbb{R}$. Notice that $\gamma \m{x}\tr\AI\m{y}$
acts as a penalty term in (\ref{BVSP}) when $\gamma \ge 0$, since
$\m{x}\tr\AI\m{y} \ge 0$ for every $\m{x},\m{y} \in \{0,1\}^n$.
Moreover, (\ref{BVSP}) gives a relaxation of (\ref{preBVSP}), since
the constraint (\ref{sepcond}) is not enforced.
Problem (\ref{BVSP}) is feasible since the VSP (\ref{VSP2}) is feasible by
assumption.
%Since $\m{w} > \m{0}$, (\ref{BVSP}) is bounded and has an
%optimal solution.
The following proposition gives conditions under which (\ref{BVSP}) is
essentially equivalent to (\ref{preBVSP}) and (\ref{VSP2}).
\smallskip

\begin{proposition} \label{binary}
If $\m{w} \ge \m{1}$ and $\gamma > 0$ with
$\gamma \ge \max\{c_i : i\in\C{V}\}$,
then for any feasible point $(\m{x},\m{y})$ in $(\ref{BVSP})$ satisfying
\begin{equation}\label{lowerbound}
f (\m{x}, \m{y}) \ge \gamma (\ell_a + \ell_b),
\end{equation}
there is a
%piecewise linear path from $(\m{x},\m{y})$ to a
feasible point
$(\bar{\m{x}},\bar{\m{y}})$ in $(\ref{BVSP})$ such that 
\begin{equation}\label{endpointconditions}
f(\bar{\m{x}}, \bar{\m{y}}) \ge f(\m{x}, \m{y})\quad \mbox{and}\quad
\bar{\m{x}}\tr \AI\bar{\m{y}} = 0.
\end{equation}
Hence, if the optimal objective value in $(\ref{BVSP})$ is at least
$\gamma (\ell_a + \ell_b)$, then there exists an optimal solution 
$(\m{x}^*, \m{y}^*)$ to $(\ref{BVSP})$ such that an optimal solution to
$(\ref{VSP2})$ is given by 
\begin{equation}\label{ABS}
\C{A} = \{ i : x^*_i = 1 \}, \;\;
\C{B} = \{ i : y^*_i = 1 \}, \;\;\mbox{and} \;\;
\C{S} = \{ i : x^*_i = y^*_i = 0 \}.
\end{equation}
\end{proposition}

\begin{proof}
Let $(\m{x},\m{y})$ be a feasible point in (\ref{BVSP}) satisfying
(\ref{lowerbound}).
Since $\m{x}$, $\m{y}$, and $\AI$ are nonnegative,
we have $\m{x}\tr\AI\m{y} \ge 0$.
If $\m{x}\tr\AI\m{y} = 0$, then we simply take
$\bar{\m{x}} = \m{x}$ and $\bar{\m{y}} = \m{y}$, and (\ref{endpointconditions})
is satisfied.
Now suppose instead that
\begin{equation}\label{pospenalty}
\m{x}\tr\AI\m{y} > 0.
\end{equation}
Then,
\begin{eqnarray}
\gamma (\ell_a + \ell_b) \le f(\m{x}, \m{y}) & = & \m{c}\tr(\m{x} + \m{y})
- \gamma\m{x}\tr\AI\m{y} \label{h67} \\
&<& \m{c}\tr (\m{x} + \m{y}) \label{h68}\\
&\le& \gamma \m{1}\tr(\m{x} + \m{y}) . \label{h69}
\end{eqnarray}
Here, (\ref{h67}) is due to (\ref{lowerbound}),
(\ref{h68}) is due to (\ref{pospenalty}) and the assumption that $\gamma > 0$,
and (\ref{h69}) holds by the assumption that
$\gamma \ge \max\{c_i : i\in\C{V}\}$.
It follows that either $\m{1}\tr\m{x} > \ell_a$ or
$\m{1}\tr\m{y} > \ell_b$.

Assume without loss of generality that $\m{1}\tr\m{x} > \ell_a$.
Since $\m{x}$ is binary and $\ell_a$ is an integer, we have
\[
\m{1}\tr\m{x} \ge \ell_a + 1.
\]
Since the entries in $\m{x}$, $\m{y}$, and $\AI$ are all non-negative
integers, (\ref{pospenalty}) implies that there exists an index $i$ such
that $\AI_i \m{y} \ge 1$ and $x_i = 1$ (recall that
subscripts on a matrix correspond to the rows).
If $\hat{\m{x}} = \m{x} - \m{e}_i$, then $(\hat{\m{x}},\m{y})$
is feasible in problem (\ref{BVSP})
since $u_a \ge \m{w}\tr\m{x} > \m{w}\tr\hat{\m{x}}$ and
\[
\m{w}\tr \hat{\m{x}} \ge \m{1}\tr\hat{\m{x}} = \m{1}\tr\m{x} - 1 \ge \ell_a.
\]
Here the first inequality is due to the assumption that $\m{w}\ge \m{1}$.
Furthermore,
\begin{equation}\label{H}
f(\hat{\m{x}},\m{y}) =
f(\m{x},\m{y}) - c_i + \gamma \AI_i \m{y}
\ge f(\m{x},\m{y}) - c_i + \gamma
\ge f(\m{x},\m{y}),
\end{equation}
since $\AI_i\m{y} \ge 1$, $\gamma \ge 0$, and $\gamma \ge c_i$.
We can continue to set components of
$\m{x}$ and $\m{y}$ to 0 until reaching a binary feasible point
$(\bar{\m{x}}, \bar{\m{y}})$ for which
$\bar{\m{x}}\tr \AI\bar{\m{y}} = 0$ and
$f(\bar{\m{x}}, \bar{\m{y}}) \ge f(\m{x}, \m{y})$.
This completes the proof of the first claim in the proposition.

Now, if the optimal objective value in (\ref{BVSP}) is at least 
$\gamma (\ell_a + \ell_b)$, then by the first part of the proposition, we
may find an
%piecewise linear path from an optimal solution of (\ref{BVSP}) to
%another
optimal solution $(\m{x}^*,\m{y}^*)$ satisfying (\ref{sepcond});
hence, $(\m{x}^*,\m{y}^*)$ is feasible in (\ref{preBVSP}).
Since (\ref{BVSP}) is a relaxation of (\ref{preBVSP}), $(\m{x}^*,\m{y}^*)$ is
optimal in (\ref{preBVSP}).
Hence, the partition $(\C{A},\C{S},\C{B})$ defined by (\ref{ABS}) is optimal
in (\ref{VSP2}). This completes the proof.
\end{proof}
\smallskip

\noindent
Algorithm~\ref{alg0} represents the procedure used in the proof of
Proposition~\ref{binary} to move from a feasible point in (\ref{BVSP}) to
a feasible point $(\bar{\m{x}},\bar{\m{y}})$ satisfying
(\ref{endpointconditions}). 

\renewcommand\figurename{Algorithm}
\begin{figure}
\begin{center}
\parbox{0cm}{\begin{tabbing}
\hspace{.4in}\=\hspace{.4in}\=\hspace{.4in}\= \kill
%{\bf Algorithm \ref{sectProblem}.1}
%({CONVERT BINARY TO SEPARATOR}) \\
{\bf Input:} A binary feasible point $(\m{x}, \m{y})$ for (\ref{BVSP}) satisfying (\ref{lowerbound})\\
{\bf while} ( $\m{x}\tr\AI\m{y} > 0$ ) \\
\>{\bf if} ( $\m{1}\tr\m{x} > \ell_a$ ) \\
\>\>Choose $i$ such that $x_i = 1$ and $\AI_i\m{y} \ge 1$.\\
\>\>Set $x_i = 0$. \\
\>{\bf else if} ( $\m{1}\tr\m{y} > \ell_b$ )\\
\>\>Choose $i$ such that $y_i = 1$ and $\AI_i\m{x}\ge 1$.\\
\>\>Set $y_i = 0$.\\
\>{\bf end if}\\
{\bf end while}
\end{tabbing}}
\end{center}
\caption{Convert a binary feasible point for $(\ref{BVSP})$ into a
vertex separator without decreasing the objective function value.
\label{alg0}}
\end{figure}
\renewcommand\figurename{Fig.}

\begin{remark}
There is typically an abundance of feasible points in $(\ref{BVSP})$ satisfying
$(\ref{lowerbound})$. For example, in the common case where
$\gamma = c_i = w_i = 1$ for each $i$, $(\ref{lowerbound})$
is satisfied whenever 
$\m{x}$ and $\m{y}$ are incidence vectors for a pair of feasible sets
$\C{A}$ and $\C{B}$ in $(\ref{VSP2})$, since in this case
\[
f (\m{x},\m{y}) = \m{c}\tr(\m{x} + \m{y}) = \m{w}\tr\m{x} + \m{w}\tr\m{y}
\ge \ell_a + \ell_b = \gamma (\ell_a + \ell_b).
\]
\end{remark}

Now consider the following continuous
bilinear program, which is obtained from (\ref{BVSP}) by relaxing the
binary constraint $\m{x},\m{y} \in \{0,1\}^n$:
\begin{eqnarray}\label{CVSP}
% objective
& \displaystyle {\max_{\m{x}, \m{y} \in \mathbb{R}^n} \;\;
f (\m{x}, \m{y}) := \; \m{c}\tr(\m{x} + \m{y}) -
\gamma\m{x}\tr\AI\m{y}} & \\
% constraints
& \begin{array}{cl} \mbox{subject to} &
\m{0}\le\m{x}\le\m{1}, \quad
\m{0}\le\m{y}\le\m{1}, \quad
\ell_a \le \m{w}\tr\m{x}  \le u_a,\;\; \mbox{ and }\;\;
\ell_b \le \m{w}\tr\m{y}  \le u_b.
\end{array}
\nonumber
\end{eqnarray}
In \cite{HagerHungerford14}, the authors study (\ref{CVSP}) in the common case
where $\m{c} \ge \m{0}$ and $\m{w} = \m{1}$.
In particular, the following theorem is proved:
%(see \cite[Theorem 2.1, Part 1]{HagerHungerford14}):
\smallskip

\begin{theorem}{\rm\cite[Theorem 2.1, Part 1]{HagerHungerford14}}
\label{cvsp}
If $(\ref{VSP2})$ is feasible,
$\m{w} = \m{1}$, $\m{c} \ge \m{0}$,
$\gamma \ge \max \{c_i : i \in \C{V} \} > 0$, and the optimal
objective value in $(\ref{VSP2})$ is at least $\gamma (\ell_a + \ell_b)$, then
$(\ref{CVSP})$ has a binary optimal solution
$(\m{x},\m{y})\in\{0,1\}^{2n}$ satisfying $(\ref{sepcond})$.
%\begin{itemize}
%\item[{\rm 1.}] The continuous program {\rm(\ref{CVSP})} has a binary optimal
%solution $(\m{x},\m{y})\in\{0,1\}^{2n}$.
%\item[{\rm 2.}] Every strict local maximizer of {\rm (\ref{CVSP})} is binary.
%\item[{\rm 3.}] Suppose $\gamma \ge \max\{c_i : i\in\C{V}\}$ and
%that the total cost $\C{C} (\C{S}^*)$ of an optimal vertex separator $\C{S}^*$
%satisfies
%\begin{equation}\label{sepbound}
%\C{C}(\C{S}^*) \le \left( \sum_{i=1}^n c_i  \right) - \gamma(\ell_a + \ell_b).
%\end{equation}
%Then there exists a binary optimal solution $(\m{x},\m{y})$ to problem
%{\rm (\ref{CVSP})} such that (\ref{sepcond}) holds.
%Moreover, if $\gamma > \max\{c_i : i\in\C{V}\}$, then every binary optimal
%solution satisfies {\rm (\ref{sepcond})}.
%\item[{\rm 4.}]
%For any binary optimal solution $(\m{x}^*,\m{y}^*)$ of $(\ref{CVSP})$
%satisfying condition {\rm (\ref{sepcond})}, an optimal solution
%to the vertex separator problem {\rm(\ref{VSP})} is given by (\ref{ABS}).
%\end{itemize}
\end{theorem}
\smallskip

In the proof of Theorem \ref{cvsp}, a step-by-step
procedure is given
for moving from any feasible point $(\m{x},\m{y})$ in (\ref{CVSP})
to a binary point $(\bar{\m{x}},\bar{\m{y}})$ satisfying
$f (\bar{\m{x}},\bar{\m{y}}) \ge f (\m{x},\m{y})$.
Thus, when the vertices all have unit weights,
the VSP may be solved with a 4-step procedure:
\begin{itemize}
\item[1.]
Obtain an optimal solution to the continuous
bilinear program (\ref{CVSP}).
\item[2.]
Move to a binary optimal solution using
the algorithm of \cite[Theorem 2.1, Part 1]{HagerHungerford14}.
\item [3.]
Convert the binary solution of (\ref{CVSP}) to a separator using
Algorithm~\ref{alg0}.
\item [4.]
Construct an optimal partition via (\ref{ABS}).
\end{itemize}

When $G$ has a small number of vertices, the dimension of the bilinear program
(\ref{CVSP})
is small, and the above approach may be very effective.
However, since the objective function in (\ref{CVSP}) is non-concave,
the number of local maximizers in (\ref{CVSP}) grows quickly
as $|\C{V}|$ becomes large and solving
the bilinear program becomes increasingly difficult.

In order to find good approximate solutions to (\ref{CVSP}) when $G$ is large,
we will incorporate the 4-step procedure (with some modifications) into
a multilevel framework (see Section \ref{sectAlgorithm}).
The basic idea is to coarsen the graph into a smaller graph having a similar
structure to the original graph; the VSP is then solved for the coarse graph
via a procedure similar to the one above,
and the solution is uncoarsened to give a solution for the original graph.

At the coarser levels of our algorithm, each vertex represents an aggregate of
vertices from the original graph. Hence, in order to keep track of the sizes of
the aggregates, weights must be assigned to the vertices in the coarse graphs,
which means the assumption of Theorem \ref{cvsp} that $\m{w} = \m{1}$ does
not hold at the coarser levels.
Indeed, when $\m{c}\ge\m{0}$ and $\m{w}\neq\m{1}$, (\ref{CVSP}) may not have
a binary optimal solution, as we will show. However, in the general case
where $\m{w} > \m{0}$ and $\m{c}\in\mathbb{R}^n$, the following weaker result
is obtained: 
%Unfortunately, when $\m{w} \neq \m{1}$, the
%formulation (\ref{CVSP}) is not necessarily exact; in fact,
%(\ref{CVSP}) may not have a binary solution, but it does have a
%mostly binary solution, which is defined as follows:
\smallskip

\begin{definition}
A point $(\m{x},\m{y}) \in \mathbb{R}^{2n}$
is called \emph{mostly binary} if $\m{x}$ and $\m{y}$ each have at most one
non-binary component.
\end{definition}
\smallskip

\begin{proposition}\label{cvspcoarse}
If the VSP {\rm (\ref{VSP2})} is feasible
and $\gamma \in \mathbb{R}$,
then $(\ref{CVSP})$
has a mostly binary optimal solution.
\end{proposition}
\smallskip

\begin{proof}
We show that the following stronger property holds:
\smallskip
\begin{itemize}
\item[(P)]
For any $(\m{x}, \m{y})$ feasible in (\ref{CVSP}), there exists a piecewise
linear path to a feasible point $(\bar{\m{x}}, \bar{\m{y}}) \in \mathbb{R}^{2n}$
which is mostly binary and satisfies
$f(\bar{\m{x}}, \bar{\m{y}}) \ge f(\m{x}, \m{y})$.
\end{itemize}
\smallskip

Let $(\m{x},\m{y})$ be any feasible point of (\ref{CVSP}).
If $\m{x}$ and $\m{y}$ each have at most one non-binary component,
then we are done.
Otherwise, assume without loss of generality there exist indices
$k \neq l$ such that
\[
0 < x_k \le x_l < 1.
\]
Since $\m{w} > \m{0}$, we can define
\[
\m{x}(t) := \m{x} + t\left(\frac{1}{w_k}\m{e}_k - \frac{1}{w_l}\m{e}_l\right)
\]
for $t\in\mathbb{R}$.
Substituting $\m{x} = \m{x}(t)$ in the objective function yields
\[
f(\m{x}(t), \m{y}) = f(\m{x}, \m{y}) + t d, \quad
\mbox{where }
d = \nabla_{\m{x}} f(\m{x},\m{y})
\left(
\frac{1}{w_k}\m{e}_k - \frac{1}{w_l}\m{e}_l \right) .
\]
If $d \ge 0$, then we may increase $t$ from zero until either $x_k(t) = 1$ or
$x_l(t) = 0$.
In the case where $d < 0$, we may decrease $t$ until
either $x_k(t) = 0$ or $x_l(t) = 1$.
In either case, the number of non-binary
components in $\m{x}$ is reduced by at least one,
while the objective value does not decrease by the choice
of the sign of $t$.
Feasibility is maintained since $\m{w}\tr\m{x}(t) = \m{w}\tr\m{x}$.
We may continue moving components to bounds in this manner
until $\m{x}$ has at most one non-binary component.
The same procedure may be applied to $\m{y}$.
In this way, we will arrive at a feasible point
$(\bar{\m{x}},\bar{\m{y}})$
such that $\bar{\m{x}}$ and $\bar{\m{y}}$ each have at most one non-binary
component and $f(\bar{\m{x}},\bar{\m{y}})\ge f(\m{x},\m{y})$.
This proves (P), which completes the proof.
\end{proof}
\smallskip

%{\it Remark \ref{cvspcoarse}.2 }
The proof of Proposition~\ref{cvspcoarse} was constructive.
A nonconstructive proof goes as follows: Since the
quadratic program (\ref{CVSP}) is bilinear, there exists an
optimal solution lying at an extreme point \cite{Konno76}.
At an extreme point of the feasible set of (\ref{CVSP}),
exactly $2n$ linearly independent constraints are active.
Since there can be at most $n$ linearly independent constraints which are
active at $\m{x}$, and similarly for $\m{y}$, there must exist exactly
$n$ linearly independent constraints which are active at $\m{x}$; in particular,
at least $n - 1$ components of $\m{x}$ must lie at a bound, and similarly for
$\m{y}$. Therefore, $(\m{x},\m{y})$ is mostly binary.
%If only $n-2$ or fewer constraints components of say $\m{x}$
%are at the bounds 0 or 1, then out of the remaining $n+2$ inequality
%constraints, at most $n+1$ are linearly independent;
%one of the balance constaints $\ell_a \le \m{w}\tr\m{x} \le u_a$
%could be active, and there could be at most $n$ linearly independent
%constaints on $\m{y} \in \mathbb{R}^n$.
%Since $(n - 2) + (n+1) = 2n -1 < 2n$, $(\m{x}, \m{y})$ cannot be an
%extreme point of the feasible set.
%If at least $n-1$ of the bound constraints are active for both
%$\m{x}$ and $\m{y}$, then $(\m{x},\m{y})$ is mostly binary.
In the case where
$\m{w} \neq \m{1}$, there may exist extreme points of the feasible set
which are not binary; for example, consider $n=3$, $\ell_a = \ell_b = 1$,
$u_a = u_b = 2$, $\m{w} = (1,1,2)$, $\m{x} = (1, 0, 0.5)$, and
$\m{y} = (0, 1, 0.5)$.

Often, the conclusion of Proposition~\ref{cvspcoarse} can be further
strengthened to assert the existence of a solution $(\m{x}, \m{y})$
of (\ref{CVSP}) for which either
$\m{x}$ or $\m{y}$ is completely binary, while the other variable has at most
one nonbinary component.
The rationale is the following:
Suppose that $(\m{x},\m{y})$ is
a mostly binary optimal solution
and without loss of generality
$x_i$ is a nonbinary component of $\m{x}$.
Substituting $\m{x}(t) = \m{x} + t\m{e}_i$ in the objective function
we obtain
\[
f(\m{x}(t), \m{y}) = f(\m{x}, \m{y}) + td, \quad
d = \nabla_{\m{x}} f(\m{x}, \m{y})\m{e}_i.
\]
If $d \ge 0$, we increase $t$, while if $d < 0$, we decrease $t$;
in either case, the objective function $f(\m{x}(t), \m{y})$ cannot decrease.
If
\begin{equation}\label{upperlowerbounds}
\ell_a + w_i \le \m{w}\tr\m{x} \le u_a - w_i,
\end{equation}
then we can let $t$ grow in magnitude until either $x_i(t) = 0$ or
$x_i(t) = 1$, while complying with the bounds
$\ell_a \le \m{w}\tr\m{x}(t) \le u_a$.
In applications, either the inequality (\ref{upperlowerbounds}) holds,
or an analogous inequality 
$\ell_b + w_j \le \m{w}\tr\m{y} \le u_a - w_j$ holds for $\m{y}$, where
$y_j$ is a nonbinary component of $\m{y}$.
The reason that one of these inequalities holds is that
we typically have $u_a = u_b > \C{W}(\C{V})/2$, which implies
that the upper bounds $\m{w}\tr\m{x} \le u_a$ and
$\m{w}\tr\m{y} \le u_b$ cannot be simultaneously active.
On the other hand, the lower bounds 
$\m{w}\tr\m{x} \ge \ell_a$ and
$\m{w}\tr\m{y} \ge \ell_b$ are often trivially satisfied when $\ell_a$ and
$\ell_b$ are small numbers like one.

Algorithm~\ref{alg4} represents the procedure used in the proof
of Proposition~\ref{cvspcoarse} to convert a given feasible point
for (\ref{CVSP}) into a mostly binary
feasible point without decreasing the objective function value.
In the case where $\m{w} = \m{1}$,
the final point returned by Algorithm \ref{alg4} is binary.
\renewcommand\figurename{Algorithm}
\begin{figure}
\begin{center}
\parbox{0cm}{\begin{tabbing}
\hspace{.4in}\=\hspace{.4in}\=\hspace{.4in}\= \kill
%{\bf Algorithm \ref{sectProblem}.2}
%({CONVERT TO MOSTLY BINARY}) \\
{\bf Input:} A feasible point $(\m{x}, \m{y})$ for the continuous bilinear
program (\ref{CVSP}). \\
{\bf while} ( $\m{x}$ has at least $2$ nonbinary components ) \\
\>Choose $i,j \in \C{V}$ such that $x_i, x_j \in (0,1)$. \\
\>Update $\m{x} \; \leftarrow \; \m{x} +
t (\frac{1}{w_i} \m{e}_i - \frac{1}{w_j} \m{e}_j)$, choosing $t$ to ensure that:
\\
\>\>(a) $f (\m{x}, \m{y})$ does not decrease, \\
\>\>(b) either $x_i \in \{0,1\}$ or $x_j \in \{0, 1\}$,\\
\>\>(c) $\m{x}$ feasible in (\ref{CVSP}). \\
{\bf end while} \\
{\bf while} ( $\m{y}$ has at least $2$ nonbinary components ) \\
\>Choose $i,j \in \C{V}$ such that $y_i, y_j \in (0,1)$. \\
\>Update $\m{y} \; \leftarrow \; \m{y} +
t (\frac{1}{w_i} \m{e}_i - \frac{1}{w_j} \m{e}_j)$,
choosing $t$ to ensure that: \\
\>\>(a) $f (\m{x}, \m{y})$ does not decrease,\\
\>\>(b) either $y_i \in \{0,1\}$ or $y_j \in \{0, 1\}$, \\
\>\>(c) $\m{y}$ feasible in (\ref{CVSP}). \\
{\bf end while}
\end{tabbing}}
\end{center}
\caption{Convert a feasible point for {\rm (\ref{CVSP})} into a mostly binary
feasible point without decreasing the objective value.
\label{alg4}}
\end{figure}
\renewcommand\figurename{Fig.}
Although the continuous bilinear problem (\ref{CVSP}) is not necessarily
equivalent
to the discrete VSP (\ref{VSP2}) when $\m{w} \ne \m{1}$,
it closely approximates (\ref{VSP2}) in the sense it has
a mostly binary optimal solution.
Since (\ref{CVSP}) is a relaxation of (\ref{preBVSP}), the objective value
at an optimal solution to (\ref{CVSP}) gives an
upper bound on the optimal objective value in (\ref{preBVSP}), and therefore
on the optimal objective value
in (\ref{VSP2}). On the other hand, given a mostly binary solution to
(\ref{CVSP}),
we can typically push the remaining
fractional components to bounds without violating the constraints on 
$\m{w}\tr\m{x}$ and $\m{w}\tr\m{y}$. Then we apply Algorithm~\ref{alg0} to
this binary point to obtain a feasible point in (\ref{preBVSP}), giving
a \emph{lower} bound on the optimal objective value in (\ref{preBVSP}) and
(\ref{VSP2}). In the case where $\m{w} = \m{1}$, the upper and lower bounds
are equal.

\section{Solving the bilinear program}
\label{sectSolve}
%-------------------------------------------------------------------------------
In our implementation of the multilevel algorithm, we use an iterative
optimization
algorithm to compute a local maximizer of
(\ref{CVSP}), then employ two different techniques to escape from a
local optimum.
Our optimization algorithm is a modified version of a
Mountain Climbing Algorithm originally proposed by
Konno \cite{Konno76} for solving bilinear programs.

%-------------------------------------------------------------------------------
\subsection{Mountain Climbing}
%-------------------------------------------------------------------------------
Given an initial guess,
the Mountain Climbing Algorithm of \cite{Konno76} solves
(\ref{CVSP}) by alternately holding
$\m{x}$ or $\m{y}$ fixed while
optimizing over the other variable.
This optimization problem for a single variable can be done efficiently
since the program is linear in $\m{x}$ and in $\m{y}$.
In our modified version of the mountain climbing algorithm,
which we call MCA (see Algorithm~\ref{alg1}),
we maximize over $\m{x}$ with $\m{y}$ held fixed to obtain
$\hat{\m{x}}$, we maximize over $\m{y}$ with $\m{x}$ held
fixed to obtain $\hat{\m{y}}$,
and then we maximize over the subspace spanned by the two maximizers.
Due to the bilinear structure of the objective function, the subspace
maximum is either
$f(\hat{\m{x}}, \hat{\m{y}})$, $f(\hat{\m{x}}, \m{y})$,
or $f(\m{x}, \hat{\m{y}})$.
The step $(\hat{\m{x}},\hat{\m{y}})$ is only taken if it provides an improvement
of at least $\eta$ more than either an $\hat{\m{x}}$ or $\hat{\m{y}}$
step, where $\eta$
is a small constant ($10^{-5}$ in our experiments). After an $\m{x}$ or $\m{y}$
step is taken, the subspace maximizer alternates between
$(\hat{\m{x}}, \m{y})$ and $(\m{x},\hat{\m{y}})$,
and hence only one linear program is solved at each iteration.
In our statement of MCA, $\C{P}_{a}$ and $\C{P}_{b}$
denote the polyhedral feasible sets for (\ref{CVSP}) defined by
\[
\C{P}_{i} = \{ \m{z} \in \mathbb{R}^n : \m{0} \le \m{z} \le \m{1} \;\;
\mbox{and} \;\; \ell_i \le \m{w}\tr\m{z} \le u_i \}, \quad i = a, b .
\]
%
%------------------------------------------------------------------------------%
% MCA %-------------------------------------------------------------------%
%------------------------------------------------------------------------------%
\renewcommand\figurename{Algorithm}
\begin{figure}[t]
\begin{center}
\parbox{0cm}{\begin{tabbing}
\hspace{.4in}\=\hspace{.4in}\=\hspace{.4in}\=\hspace{.4in} \kill
%{\bf Algorithm \ref{sectSolve}.1} ({MCA}) \\
{\bf Input:} A feasible point $({\m{x}}, {\m{y}})$ for (\ref{CVSP}) and
$\eta > 0$.\\
{\bf while} ( $(\m{x}, \m{y})$ not stationary point for (\ref{CVSP}) ) \\
\>$\hat{\m{x}} \; \leftarrow \; \mbox{argmax} \;
\{ f(\m{x}, {\m{y}} ) : \m{x} \in \C{P}_a\}$\\
\>$\hat{\m{y}} \; \leftarrow \; \mbox{argmax} \;
\{ f({\m{x}}, \m{y} ) : \m{y} \in \C{P}_b \}$ \\
\>{\bf if} ( $f (\hat{\m{x}}, \hat{\m{y}}) > \max \; \{ f (\hat{\m{x}},
{\m{y}}), f({\m{x}}, \hat{\m{y}}) \} + \eta$ ) \\
\>\>$({\m{x}}, {\m{y}}) \; \leftarrow \; (\hat{\m{x}}, \hat{\m{y}})$\\
\>{\bf else if} ( $f (\hat{\m{x}}, {\m{y}}) > f ({\m{x}}, \hat{\m{y}})$ )\\
\>\>${\m{x}} \; \leftarrow \; \hat{\m{x}}$\\
\>{\bf else} \\
\>\>$\m{y} \; \leftarrow \; \hat{\m{y}}$\\
\>{\bf end if}\\
{\bf end while}\\
{\bf return} $(\m{x},\m{y})$
\end{tabbing}}
\end{center}
\caption{{\bf MCA:}
A modified version of Konno's Mountain Climbing Algorithm for
generating a stationary point for
$(\ref{CVSP})$.
\label{alg1}}
\end{figure}
\renewcommand\figurename{Fig.}

The linear programs arising in MCA are solved using a greedy algorithm.
In particular, $\max\;\{f(\m{x},\m{y}) : \m{x}\in\C{P}_a\}$ is solved by
setting all components $x_i$ equal to zero and then
sorting the components in decreasing order of their ratio
$r_i := \frac{\partial f}{\partial x_i}(\m{x},\m{y})/{w_i}$;
the components are then visited in
order and are pushed up from $0$ to $1$ until
either $\m{w}\tr\m{x} = u_a$ or a component is reached such that $r_i < 0$.
It is easy to show that this procedure gives an optimal solution to the LP.

%The iterates of MCA
%typically converge to a stationary point satisfying (V1) (or (C1) and (C2))
%in only a few iterations.
%The user may then optionally check
%the second-order optimality conditions (V2)--(V5); if any of these
%conditions fail to hold,
%then Table~\ref{ascent} provides a direction of ascent;
%after moving to a new, strictly better point, MCA can be restarted.
%Eventually, we reach a local maximizer of (\ref{CVSP}).

Since (\ref{CVSP}) is non-concave, it may have many stationary points;
thus, it is crucial to incorporate techniques to escape
from a stationary point and explore a new part of the solution space.
The techniques we develop are based on perturbations in
%and potentially reach a better maximizer of (\ref{CVSP}),
%our strategy is to perturb the original problem
%(\ref{CVSP}) until the current local maximizer is no longer a stationary point.
either the cost vector $\m{c}$ or in the penalty parameter $\gamma$.
We make the smallest possible perturbations which guarantee that the
current iterate is no longer a stationary point of the perturbed problem.
After computing an approximate solution of the perturbed problem, we use
it as a starting guess in the original problem and reapply MCA.
If we reach a better solution for the original problem, then we save this best
solution, and make another perturbation.
If we do not reach a better solution, then we continue the perturbation
process, starting from the new point.

%-------------------------------------------------------------------------------
\subsection{$\m{c}$-perturbations}
%-------------------------------------------------------------------------------
Our perturbations of the cost vector are based on an analysis of the
first-order optimality conditions for the maximization problem (\ref{CVSP}).
If a feasible point $(\m{x}, \m{y})$ for (\ref{CVSP}) is a local maximizer,
then the objective function can only decrease when we make small
moves in the direction
of other feasible points, or equivalently,
\begin{equation}\label{1}
\nabla_{\m{x}} f(\m{x}, \m{y}) (\tilde{\m{x}} - \m{x}) +
\nabla_{\m{y}} f(\m{x}, \m{y}) (\tilde{\m{y}} - \m{y}) \le 0
\end{equation}
whenever $(\tilde{\m{x}}, \tilde{\m{y}})$ is feasible in (\ref{CVSP}).
Another way to state the first-order optimality condition (\ref{1}) employs
multipliers for the constraints.
In particular, by \cite[Theorem 12.1]{NocedalWright2006},
a feasible point $(\m{x}, \m{y})$ for (\ref{CVSP}) is a local maximizer only
if there exist multipliers
$\g{\mu}^a \in \C{M}(\m{x})$,
$\g{\mu}^b \in \C{M}(\m{y})$,
$\lambda^a \in \C{L}(\m{x}, \ell_a, u_a)$,
$\lambda^b \in \C{L}(\m{y}, \ell_b, u_b)$
such that
\begin{equation}\label{sep-1st}
\left[
\begin{array}{c}
\nabla_{\m{x}} f(\m{x},\m{y})\\
\nabla_{\m{y}} f(\m{x},\m{y})
\end{array}
\right]
+
\left[
\begin{array}{c}
\g{\mu}^a\\
\g{\mu}^b
\end{array}
\right]
+
\left[
\begin{array}{c}
\lambda^a\m{w}\\
\lambda^b\m{w}
\end{array}
\right]
= \m{0} ,
\end{equation}
where
\begin{eqnarray*}
\C{M}(\m{z}) &=& \{\g{\mu}\in\mathbb{R}^n :  \mu_i z_i \le \min\{\mu_i,0\}
\;\mbox{ for all } \; 1 \le i \le n\} \quad \mbox{and}\\
\C{L}(\m{z},\ell,u) &=& \{\lambda\in\mathbb{R} : \lambda\m{w}\tr\m{z} \le
\min\{\lambda u,\lambda \ell\}\}.
\end{eqnarray*}
The conditions (\ref{1}) and (\ref{sep-1st}) are equivalent.
The condition (\ref{sep-1st}) is often called the KKT (Karush-Kuhn-Tucker)
condition.
The usual formulation of the KKT conditions involves introducing
a separate multiplier for each upper and lower bound constraint, which leads to
eight different multipliers in the case of (\ref{CVSP}).
In (\ref{sep-1st}) the number of multipliers has been reduced to four through
the use of the set $\C{M}$ and $\C{L}$.

In describing our perturbations to the objective function,
we attach a subscript to $f$ to indicate the parameter that is being perturbed.
Thus $f_{\m{c}}$ denotes the original objective in (\ref{CVSP}),
while $f_{\tilde{\m{c}}}$ corresponds to the objective obtained by replacing
$\m{c}$ by $\tilde{\m{c}}$.
\smallskip

\begin{proposition}\label{prop_cperts}
If $\gamma \in \mathbb{R}$
and $(\m{x}, \m{y})$ satisfies the first-order optimality condition
{\rm (\ref{sep-1st})} for $(\ref{CVSP})$,
then in any of the following cases, for any choice of $\epsilon > 0$
and for the
indicated choices of $\tilde{\m{c}}$, $(\m{x}, \m{y})$
does not satisfy the first-order optimality condition
{\rm (\ref{sep-1st})} for $f = f_{\tilde{\m{c}}}$:
\begin{itemize}
\item[{\rm 1.}]
For any $i \ne j$ such that $\mu_i^a = \mu_j^a = 0$, $x_i < 1$,
and $x_j > 0$, take
\[
\tilde{c}_k = \left\{
\begin{array}{cl}
c_k + \epsilon & \mbox{if } k = i, \\
c_k - \epsilon & \mbox{if } k = j, \\
c_k            & \mbox{otherwise} 
\end{array} \right. .
\]
\item[{\rm 2.}]
If $\lambda^a = 0$ and $\m{w}\tr \m{x} < u_a$, then for any
$i$ such that $\mu_i^a = 0$ and $x_i < 1$, take
$\tilde{c}_i = c_i + \epsilon$
and $\tilde{c}_k = c_k$ for $k \ne i$.
\item[{\rm 3.}]
If $\lambda^a = 0$ and $\m{w}\tr \m{x} > \ell_a$, then for any
$j$ such that $\mu_j^a = 0$ and $x_j > 0$, take
$\tilde{c}_j = c_j - \epsilon$ and
$\tilde{c}_k = c_k$ for $k \ne j$.
\end{itemize}
\end{proposition}
\smallskip

\begin{proof}
{\bf Part 1.}
Let $i$ and $j$ satisfy the stated conditions and
define the vector $\m{d} = w_j \m{e}_i - w_i \m{e}_j$.
Since $\m{w}\tr \m{d} = 0$,
$x_i < 1$, and $x_j > 0$,
it follows that $\tilde{\m{x}}(t) = \m{x} + t\m{d}$ is feasible in
(\ref{CVSP}) for $t > 0$ sufficiently small.
Moreover, by {\rm (\ref{sep-1st})}
and the assumption $\mu_i^a = \mu_j^a = 0$,
we have $\nabla_{\m{x}} f_{\m{c}} (\m{x}, \m{y}) \m{d} = 0$.
Since
\[
\nabla_{\m{x}} f_{\tilde{\m{c}}} (\m{x}, \m{y}) =
\nabla_{\m{x}} f_{{\m{c}}} (\m{x}, \m{y}) + \epsilon (\m{e}_i\tr - \m{e}_j\tr),
\]
we conclude that
$\nabla_{\m{x}} f_{\tilde{\m{c}}} (\m{x}, \m{y})\m{d} = \epsilon (w_i + w_j)$,
which implies that
\begin{equation}\label{eh}
\nabla_{\m{x}} f_{\tilde{\m{c}}} (\m{x}, \m{y})(\tilde{\m{x}}(t) - \m{x}) =
\epsilon t (w_i + w_j) > 0.
\end{equation}
This shows that the first-order optimality condition (\ref{1}) is
not satisfied at $(\m{x}, \m{y})$ for $f = f_{\tilde{\m{c}}}$.

{\bf Part 2.}
Define $\m{d} = \m{e}_i$.
Since $\m{w}\tr\m{x} < u^a$ and $x_i < 1$, it follows that
$\tilde{\m{x}}(t) = \m{x} + t\m{d}$ is feasible
in (\ref{CVSP}) for $t > 0$ sufficiently small.
Since $\lambda^a = \mu_i^a = 0$,
{\rm (\ref{sep-1st})}
implies that
$\nabla_{\m{x}} f_{\m{c}} (\m{x}, \m{y}) \m{d} = 0$.
So, 
\[
\nabla_{\m{x}} f_{\tilde{\m{c}}} (\m{x},\m{y})\m{d} =
[\nabla_{\m{x}} f_{\m{c}}(\m{x}, \m{y}) + \epsilon\m{e}_i\tr] \m{d}
= 0 + \epsilon \m{e}_i\tr\m{d} = \epsilon,
\]
which implies $\nabla_{\m{x}} f_{\tilde{\m{c}}} (\m{x}, \m{y})
(\tilde{\m{x}} (t) - \m{x}) = \epsilon t > 0$.
Hence, the first-order optimality conditions (\ref{1}) are
not satisfied at $(\m{x}, \m{y})$ for $f_{\tilde{\m{c}}}$.

{\bf Part 3.}
The analysis parallels the analysis of Part 2.
\end{proof}
\smallskip

Of course, Proposition \ref{prop_cperts} may be applied to either $\m{x}$ or
$\m{y}$.
Since $\epsilon$ was arbitrary,
we usually take $\epsilon$ to be a tiny positive number ($10^{-6}$ in our
experiments).
By making a tiny change in the problem, the
iterates of the optimization algorithm MCA applied to $f = f_{\tilde{\m{c}}}$
must move away from the current
point to a new vertex of the feasible set to improve the objective value
in the perturbed problem.
Then the solution of the slightly perturbed problem is used as a starting
guess for the solution of the original unperturbed problem.
Algorithm~\ref{alg2}, also denoted MCA{\underscore}CP, incorporates the
$\m{c}$-perturbations into MCA. In the figure, the notation
MCA $(\m{x},\m{y},\tilde{\m{c}})$
indicates that the MCA algorithm is applied to the point $(\m{x},\m{y})$ using
$\tilde{\m{c}}$ in place of $\m{c}$ as the vector of vertex costs.
In our experiments,
the c-perturbations were performed in the following
way:
For each $i$ such that
$|{\mu}^a_i| < 10^{-5}$, we set $\tilde{c}_i = c_i + \epsilon$ whenever
$x_i < 0.5$ and $\tilde{c}_i = c_i - \epsilon$ otherwise; similar perturbations
are made based on the values of $\mu^b_i$ and $y_i$.

%%%%%%%%%%INSERT ALG2 FIGURE HERE%%%%%%%%%%%%%%%%%%%
\renewcommand\figurename{Algorithm}
\begin{figure}[t]
\begin{center}
\parbox{0cm}{
\begin{tabbing}
\hspace{.4in}\=\hspace{.4in}\=\hspace{.4in}\=\hspace{.4in} \kill
%{\bf Algorithm \ref{sectSolve}.2} ({MCA{\underscore}CP}) \\
{\bf Input:} A feasible point $({\m{x}}, {\m{y}})$ for (\ref{CVSP}).\\
$({\m{x}}, {\m{y}})\; \leftarrow\; \mbox{MCA}\; ({\m{x}}, {\m{y}})$\\
{\bf loop}\\
\>$\tilde{\m{c}}\;\leftarrow\;\mbox{perturb}\; (\m{c})$\\
\>$(\tilde{{\m{x}}}, \tilde{{\m{y}}})\;\leftarrow\;
\mbox{MCA}\;({\m{x}}, {\m{y}}, \tilde{\m{c}})$\\
\>$({\m{x}}^*, {\m{y}}^*)\;\leftarrow\;
\mbox{MCA}\;(\tilde{\m{x}}, \tilde{\m{y}}, \m{c})$\\
\>{\bf if} ( $f ({\m{x}}^*, {\m{y}}^*) > f ({\m{x}}, {\m{y}})$ ) \\
\>\>$({\m{x}}, {\m{y}})\; \leftarrow\; ({\m{x}}^*, {\m{y}}^*)$\\
\>{\bf else}\\
\>\>{\bf break}\\ 
\>{\bf end if}\\
{\bf end loop}\\
{\bf return} $(\m{x},\m{y})$
\end{tabbing}}
\end{center}
\caption{{\bf MCA{\underscore}CP:}
A modification of MCA which incorporates $\m{c}$-perturbations.
\label{alg2}}
\end{figure}
\renewcommand\figurename{Fig.}

%-------------------------------------------------------------------------------
\subsection{$\gamma$-perturbations}
%-------------------------------------------------------------------------------
Next, we consider perturbations in the parameter $\gamma$.
According to our theory, we need to take
$\gamma \ge \max\{c_i : i \in \C{V}\}$ to ensure an (approximate) equivalence
between the discrete (\ref{VSP2}) and the continuous (\ref{CVSP}) VSP.
The penalty term $-\gamma \m{x}\tr \AI\m{y}$ in the objective
function of (\ref{CVSP}) enforces the constraints
$\C{A} \cap \C{B} = \emptyset$ and
$(\C{A}\times \C{B}) \cap \C{E} = \emptyset$ of (\ref{VSP2}).
Thus, by decreasing $\gamma$, we relax our enforcement of these constraints
and place greater emphasis on the cost of the separator.
The next proposition will determine the amount by which we must
decrease $\gamma$ in order to ensure that the current point $(\m{x}, \m{y})$,
a local maximizer of $f_\gamma$, is no longer a local maximizer
of the perturbed problem $f_{\tilde{\gamma}}$.
The derivation requires a formulation of the second-order necessary
and sufficient optimality conditions given in
\cite[Cor. 3.3]{HagerHungerford13}; applying these conditions to the
bilinear program (\ref{CVSP}), we have the following theorem.
\smallskip
\begin{theorem}\label{firstsecond-order}
If $\gamma \in \mathbb{R}$ and $(\m{x}, \m{y})$ is feasible in
$(\ref{CVSP})$, then
$(\m{x}, \m{y})$ is a local maximizer if and only if
the following hold:
\begin{itemize}
\item [{\rm (C1)}]
$\nabla_{\m{x}} f (\m{x}, \m{y})
\m{d} \le 0$ for every
$\m{d} \in \C{F}_a(\m{x}) \cap \C{D}$,
\item [{\rm (C2)}]
$\nabla_{\m{y}} f (\m{x}, \m{y})
\m{d} \le 0$ for every
$\m{d} \in \C{F}_b(\m{y}) \cap \C{D}$, and
\item [{\rm (C3)}]
$\m{d}_1\tr(\nabla^2 f)\m{d}_2 \le 0$ for every
$\m{d}_1, \m{d}_2\in \C{C}(\m{x},\m{y})\cap\C{G}$,
\end{itemize}
where
\begin{eqnarray}
\C{F}_i (\m{z}) & = &
\left\{
\m{d} \in \mathbb{R}^n :
\begin{array}{l}
{d_j \le 0 \mbox{ for all } j
\mbox{ such that } z_j = 1}\\[.05in]
{d_j \ge 0 \mbox{ for all } j
\mbox{ such that } z_j = 0}\\[.05in]
\m{w}\tr\m{d} \le 0 \mbox{ if } \m{w}\tr\m{z} = u_i\\[.05in]
\m{w}\tr\m{d} \ge 0 \mbox{ if } \m{w}\tr\m{z} = l_i
\end{array}
\right\},
\quad i = a, b, \; \m{z} \in \mathbb{R}^n,\nonumber\\
\C{C}(\m{x},\m{y}) & = &
\left\{\m{d}\in\C{F}_a (\m{x})\times\C{F}_b(\m{y}) :
\nabla f (\m{x},\m{y}) \m{d} = 0\right\},\nonumber\\
\C{D} & = & \bigcup_{i,j = 1}^n\left\{\m{e}_i, -\m{e}_i,
{w_j}\m{e}_i - {w_i}\m{e}_j\right\},\quad \mbox{and}
\quad\C{G} \; = \; (\C{D}\times\{\m{0}\})\cup(\{\m{0}\}\times\C{D}).
\label{edge-description}
\end{eqnarray}
\end{theorem}
\smallskip

The sets $\C{F}_a$ and $\C{F}_b$ are the cones of first-order feasible
directions at $\m{x}$ and $\m{y}$.
The set $\C{G}$ is a reflective edge description of the feasible set,
introduced in \cite{HagerHungerford13};
that is,
each edge of the constraint polyhedron of (\ref{CVSP})
is parallel to an element of $\C{G}$.
Since $\C{D}$ is a finite set, checking the first-order
optimality conditions reduces to testing the
conditions (C1) and (C2) for
the finite collection of
elements from $\C{D}$ that are in the cone of first-order feasible directions;
testing the second-order optimality conditions reduces to testing the
condition (C3) for the elements from $\C{G}$ that are in the critical cone
$\C{C}(\m{x},\m{y})$.

\begin{proposition}\label{propgammareds}
Let $\gamma \in \mathbb{R}$,
let $(\m{x}, \m{y})$ be a feasible point in 
$(\ref{CVSP})$ which satisfies the first-order optimality condition {\rm (C1)},
and let $\tilde{\gamma}\le\gamma$.
\begin{itemize}
\item[{\rm 1.}]
Suppose $\ell_a < \m{w}\tr\m{x} < u_a$.
Then {\rm (C1)} holds at
$(\m{x}, \m{y})$ for $f = f_{\tilde{\gamma}}$
if and only if $\tilde{\gamma} \ge \alpha_1$, where
\begin{eqnarray*}
\alpha_1 := \max\;\left\{\frac{c_j}{\AI_j\m{y}} : j \in \C{J}\right\}
& \;\;\mbox{and}\;\; &
\C{J} := \left\{ j : x_j < 1 \mbox{ and } \AI_j \m{y} > 0 \right\},
\end{eqnarray*}
with $\alpha_1 := -\infty$ if $\C{J} = \emptyset$.

\item[{\rm 2.}]
Suppose $\m{w}\tr\m{x} = u_a$.
Then {\rm (C1)} holds at $(\m{x},\m{y})$
for $f = f_{\tilde{\gamma}}$
if and only if $\tilde{\gamma} \ge \alpha_2$
where
\begin{eqnarray*}
\alpha_2 := \inf\; \Gamma
& \;\;\mbox{and}\;\; &
\Gamma := \left\{ \alpha \in \mathbb{R} :
\frac{1}{w_i}\frac{\partial f_{\alpha}}{\partial x_i}(\m{x},\m{y})
\le
\frac{1}{w_j}\frac{\partial f_{\alpha}}{\partial x_j}(\m{x},\m{y})
\mbox{ for all } x_i < 1 \mbox{ and } x_j > 0 \right\}.
\end{eqnarray*}
\end{itemize}
\end{proposition}

\begin{proof}
{\bf Part 1.}
Since $\ell_a < \m{w}\tr\m{x} < u_a$,
the cone of first-order feasible directions for $\m{x}$ is given by
\[
\C{F}_a (\m{x}) =
\left\{\m{d} \in \mathbb{R}^n:
d_i \ge 0 \mbox{  when  } x_i = 0 \mbox{ and }
d_i \le 0 \mbox{  when  } x_i = 1, \; i = 1, \ldots, n
\right\}.
\]
It follows that for each $i = 1,\ldots,n$,
\[
\begin{array}{ccl}
\m{e}_i & \in & \C{F}_a (\m{x}) \;\mbox{ if and only if } \; x_i < 1, \\
-\m{e}_i & \in & \C{F}_a (\m{x}) \;\mbox{ if and only if }\; x_i > 0, \\
(w_j\m{e}_i -w_i\m{e}_j) & \in & \C{F}_a (\m{x}) \;\mbox{ if and only if } \;x_i < 1
\mbox{ and } x_j > 0.
\end{array}
\]
Hence, the first-order optimality condition (C1) for $f_{\tilde{\gamma}}$
can be expressed as follows:
\begin{eqnarray}
\nabla_{\m{x}} f_{\tilde{\gamma}} (\m{x}, \m{y}) \m{e}_i \le 0
& \; \mbox{ when } \; & x_i < 1, \label{xcond1} \\
\nabla_{\m{x}} f_{\tilde{\gamma}} (\m{x}, \m{y}) \m{e}_i \ge 0
& \; \mbox{ when } \; & x_i > 0, \label{xcond2} \\
\nabla_{\m{x}} f_{\tilde{\gamma}} (\m{x}, \m{y}) (w_j\m{e}_i - w_i\m{e}_j) \le 0
& \; \mbox{ when } \; & x_i < 1 \mbox{ and } x_j > 0. \label{xcond3}
\end{eqnarray}
Since (\ref{xcond3}) is implied by (\ref{xcond1}) and (\ref{xcond2}), it
follows that (C1) holds if and only if (\ref{xcond1}) and (\ref{xcond2}) hold.
Since (C1) holds for $f_\gamma$, we know that
\[
\nabla_{\m{x}} f_{{\gamma}} (\m{x}, \m{y}) \m{e}_i =
c_i - \gamma \AI_i\m{y} \ge 0 \;\;
\mbox{when}\;\; x_i > 0.
\]
Hence, since $\tilde{\gamma} \le \gamma$ and $\AI_i\m{y} \ge 0$,
\[
\nabla_{\m{x}} f_{\tilde{\gamma}} (\m{x}, \m{y}) \m{e}_i =
c_i - \tilde{\gamma} \AI_i\m{y} \ge 0 \quad
\mbox{when } x_i > 0.
\]
Hence, (C1) holds with respect to $\tilde{\gamma}$
if and only if (\ref{xcond1}) holds.
Since (C1) holds for $f = f_{\gamma}$, we have
\begin{equation}\label{h70}
c_j - \gamma\AI_j \m{y} \le 0 \mbox{  when  } x_j < 1.
\end{equation}
Hence, for every $j$ such that $x_j < 1$ and $\AI_j\m{y} = 0$, we
have
\[
c_j - \tilde{\gamma}\AI_j\m{y} = c_j = c_j -
\gamma\AI_j\m{y} \le 0.
\]
So, (\ref{xcond1}) holds if and only if
$c_j - \tilde{\gamma}\AI_j\m{y} \le 0$ for every
$j \in \C{J}$; that is, if and only if $\tilde{\gamma} \ge \alpha_1$.
This completes the proof of Part 1.

{\bf Part 2.}
Since $\m{w}\tr\m{x} = u_a$, the cone of first-order
feasible directions at $\m{x}$ has the constraint $\m{w}\tr\m{d} \le 0$.
Consequently, $\m{e}_i \not\in \C{F}_a (\m{x}) \cap \C{D}$ for any $i$,
and the first-order optimality condition (C1) for $f_{\tilde{\gamma}}$ reduces
to (\ref{xcond2})--(\ref{xcond3}).
As in Part 1, (\ref{xcond2}) holds, since $\tilde{\gamma} \le \gamma$.
Condition (\ref{xcond3}) is equivalent to $\tilde{\gamma} \in \Gamma$.
Since (\ref{xcond3}) holds
for $f = f_{\gamma}$, we have $\gamma \in \Gamma$.
Since $\nabla_{\m{x}} f_{\tilde{\gamma}} (\m{x}, \m{y})$ is a affine
function of $\tilde{\gamma}$, the set of $\tilde{\gamma}$ satisfying
(\ref{xcond3}) for some $i$ and $j$ such that $x_j > 0$ and $x_i < 1$
is a closed interval, and the intersection of the intervals
over all $i$ and $j$ for which $x_j > 0$ and $x_i < 1$ is also a closed
interval.
Hence, since $\tilde{\gamma} \le \gamma \in \Gamma$,
we have $\tilde{\gamma} \in \Gamma$
if and only if $\tilde{\gamma} \ge \alpha_2$.
This completes the proof of Part 2.
\end{proof}
\smallskip

{\it Remark \ref{sectSolve}.1}:
Of course, Proposition \ref{propgammareds} also holds when
the variables $\m{x}$ and $\m{y}$ and the bounds $(\ell_a, u_a)$
and $(\ell_b, u_b)$ are interchanged.
In most applications, $u_a$ and $u_b > {\C{W}(\C{V})}/{2}$,
$\ell_a = \ell_b = 1$, and as the iterates converge to a solution of
(\ref{CVSP}), either the constraint
$\m{w}\tr \m{x} \le u_a$ or the constraint $\m{w}\tr \m{y} \le u_b$ is active.
Thus, for a given iterate $(\m{x}, \m{y})$, the assumptions of Part~1
typically apply to either $\m{x}$ or $\m{y}$,
while the assumptions of Part~2 apply to the other variable.
%The first-order optimality conditions hold if and only if
%$\tilde{\gamma}$ satisfies both conditions, one obtained from Part~1 and
%the other obtained from Part~2.
Although the Part~2 condition seems complex, it often provides no useful
information in the following sense:
In a multilevel implementation, the vertex costs (and weights) are often 1 at
the finest level, and at coarser levels, the vertex costs may not differ
greatly.
When the vertex costs are equal, (\ref{xcond3})
holds when $\tilde{\gamma}$ has the same sign as $\gamma$; that is,
as long as $\tilde{\gamma} \ge 0$.
% (recall that we must take $\gamma \ge 0$
%in order to satisfy the conditions of Proposition \ref{binary}).
Thus, $\alpha_2 = 0$.
Since $\alpha_1$ is typically positive, the tighter bound on
$\tilde{\gamma}$ is the interval $[\alpha_1, \gamma]$,
which means that when $\tilde{\gamma} < \alpha_1$,
$(\m{x}, \m{y})$ is no longer a
stationary point for $f = f_{\tilde{\gamma}}$.
\smallskip

Algorithm~\ref{alg3}, also denoted MCA{\underscore}GR, approximately solves
(\ref{CVSP}), while incorporating both $\m{c}$-perturbations and
$\gamma$-refinements.
Here, the notation MCA$(\m{x},\m{y},\tilde{\gamma})$
indicates that the MCA algorithm is applied to the point $(\m{x},\m{y})$ using
$\tilde{\gamma}$ in place of $\gamma$ as the penalty parameter.
In our implementation, the $\gamma$-refinements
are performed in the following way: $\gamma$ is reduced from the initial
value $\alpha_1$ in 10 uniform decrements until it reaches 0 or the optimal
objective value improves.
We note that in \ref{alg3},
the $\m{c}$-perturbations are embedded in the $\gamma$-refinement procedure
in order to obtain a local optimizer of high quality for each $\tilde{\gamma}$.

%%%%%%%%%%INSERT ALG3 FIGURE HERE%%%%%%%%%%%%%%%%%%%
\renewcommand\figurename{Algorithm}
\begin{figure}[t]
\begin{center}
\parbox{0cm}{
\begin{tabbing}
\hspace{.4in}\=\hspace{.4in}\=\hspace{.4in}\=\hspace{.4in} \kill
%{\bf Algorithm \ref{sectSolve}.3} ({MCA{\underscore}GR}) \\
{\bf Input:} A feasible point $({\m{x}}, {\m{y}})$ for (\ref{CVSP}).\\
$({\m{x}}, {\m{y}})\; \leftarrow\; \mbox{MCA}{\underscore}\mbox{CP}\;
({\m{x}}, {\m{y}})$\\
$\tilde{\gamma}\;\leftarrow\;\alpha_1$\\
{\bf while} ( $\tilde{\gamma} > 0$ )\\
\>$\tilde{\gamma}\;\leftarrow\;\mbox{reduce}\; (\tilde{\gamma})$\\
\>$(\tilde{{\m{x}}}, \tilde{{\m{y}}})\;\leftarrow\;
\mbox{MCA}{\underscore}\mbox{CP}\;({\m{x}}, {\m{y}}, \tilde{\gamma})$\\
\>$({\m{x}}^*, {\m{y}}^*)\;\leftarrow\;
\mbox{MCA}{\underscore}\mbox{CP}\;
(\tilde{\m{x}}, \tilde{\m{y}}, \gamma)$\\
\>{\bf if} ( $f ({\m{x}}^*, {\m{y}}^*) > f ({\m{x}}, {\m{y}})$ ) \\
\>\>$({\m{x}}, {\m{y}})\; \leftarrow\; ({\m{x}}^*, {\m{y}}^*)$\\
\>\>$\tilde{\gamma}\;\leftarrow\;\alpha_1$\\
{\bf end while}\\
{\bf return} $(\m{x},\m{y})$
\end{tabbing}}
\end{center}
\caption{{\bf MCA{\underscore}GR:} A refinement algorithm for $(\ref{CVSP})$
which incorporates $\m{c}$-perturbations and $\gamma$-refinements.
\label{alg3}}
\end{figure}
\renewcommand\figurename{Fig.}

\section{Multilevel algorithm}
\label{sectAlgorithm}
%-------------------------------------------------------------------------------
We now give an overview of a multilevel algorithm, which we call BLP, for
solving the vertex separator problem.
The algorithm consists of three phases: coarsening, solving, and
uncoarsening.

{\bf Coarsening.}
Vertices are visited one by one and each vertex is matched with an unmatched
adjacent vertex, whenever one exists. Matched vertices are merged together
to form a single vertex having a cost and weight equal to the sum of the
costs and weights of the constituent vertices. Multiple edges which arise
between two vertices are combined
and assigned an edge weight equal to the sum of the weights of combined edges
(in the original graph, all edges are assumed to have weight $1$). This
coarsening process repeats until the graph has fewer than $75$ vertices or
fewer than $10$ edges.

The goal of the coarsening phase is to reduce the number
of degrees of freedom in the problem, while preserving its structure so that
the solutions obtained for the coarse problems give a good approximation
to the solution for the original problem.
We consider two matching rules: random and heavy-edge.
In heavy-edge based matching, 
each vertex is matched with an unmatched neighbor such that
the edge between them has the greatest weight over all unmatched neighbors.
Heavy edge matching rules have
been used in multilevel algorithms such as
\cite{Hendrickson, KarypisKumar98e}, and were originally developed for edge-cut
problems.
%Thus, during the coarsening phase, it is
%important that we only match vertices which are likely to appear in the
%same set in the optimal partition for the original graph.
%In our experiments,
%we consider two matching rules which attempt to measure this likelihood.
%First, we consider the standard \emph{heavy edge} matching rules,
%in which each vertex is matched with an unmatched neighbor such that
%the edge between them has the greatest weight over all unmatched neighbors.
%Heavy edge matching rules have
%been used in multilevel algorithms such as
%\cite{Hendrickson, KarypisKumar98e}.
In our initial experiments, we also considered a third rule based on
an \emph{algebraic distance} \cite{Safro11} between vertices. However, 
the results were not significantly different from heavy-edge matching, which
is not surprising, since (like heavy-edge rules)
the algebraic distance was originally developed for minimizing edge-cuts.
%We will also consider a matching rule based on an \emph{algebraic distance}
%\cite{Safro11} between vertices. In particular, each vertex is matched with an
%unmatched neighbor for which the algebraic distance to that neighbor is the
%smallest out of all unmatched neighbors. The algebraic distance
%on $G$ is computed as follows:
%Let $\m{L} = diag(\m{A}\m{1}) - \m{A}$ be the graph Laplacian, where
%$\m{A}$ is the $n \times n$ unweighted adjacency matrix for $G$.
%Let $\{\m{x}^{(0,m)}\}_{m = 1}^r \subseteq \mathbb{R}^n$
%be $r$ randomly initialized test vectors.
%The ($2$-normed) algebraic distance between $i, j \in \C{V}$ in $G$ is defined
%as $d_{i,j} = 1/\rho_{i,j}$, where
%\[
%\rho_{i,j} = \sum_{m = 1}^r |x_i^{(K,m)} - x_j^{(K,m)}|^2.
%\]
%Here, $\m{x}^{(K,m)}$ is the $k$-th iterate when Jacobi over-relaxation is
%applied to the test vector $\m{x}^{(0,m)}$ for the system
%$\m{Lx} = \m{0}$. A development of the algebraic distance, as
%well as arguments for its usefulness in graph partitioning is given in
%\cite{Safro11}.

{\bf Solving.}
For each of the graphs in the multilevel hierarchy, 
we approximately solve
(\ref{CVSP}) using MCA{\underscore}GR.
%At any level, we check the
%second-order optimality conditions and use the ascent directions of
%Table~\ref{ascent} to escape from a stationary point
%which is not a local maximizer.
For the coarsest graph, the starting guess is $x_i = u_a/\C{W}(\C{V})$ and
$y_i = u_b/\C{W}(\C{V})$, $i = 1,2,\ldots,n$.
For the finer graphs, a starting guess is obtained from the next coarser
level using the uncoarsening process described below.
After MCA{\underscore}GR terminates,
Algorithm~\ref{sectProblem}.2 along with the
modification discussed after Proposition~\ref{cvspcoarse} are used to
obtain a binary approximation to a solution of (\ref{CVSP}), and then
Algorithm~\ref{sectProblem}.1 is used to convert the binary solution
into a vertex separator.

{\bf Uncoarsening.}
We use the solution for the vertex separator problem computed at any
level in the multilevel hierarchy
as a starting guess for the solution at the next finer level.
Sophisticated cycling techniques like the W- or F-cycle
\cite{vlsicad} were not implemented.
A starting guess for the next finer graph is obtained by unmatching
vertices in the coarser graph.
Suppose that we are uncoarsening from level $l$ to $l-1$ and
$(\m{x}^l, \m{y}^l)$ denotes the solution computed at level $l$.
If vertex $i$ at level $l$ is obtained by matching vertices $j$ and $k$
at level $l-1$, then our starting guess for $(\m{x}^{l-1}, \m{y}^{l-1})$
is simply $(x_{j}^{l-1}, y_{j}^{l-1}) = (x_{i}^{l}, y_{i}^{l})$ and
$(x_{k}^{l-1}, y_{k}^{l-1}) = (x_{i}^{l}, y_{i}^{l})$.

%-------------------------------------------------------------------------------
\section{Numerical results}
\label{sectResults}
%-------------------------------------------------------------------------------
The multilevel algorithm was programmed in C++ and compiled using g++ with
optimization O3 on a Dell
Precision T7610 Workstation with a Dual Intel Xeon Processor
E5-2687W v2 (16 physical cores, 3.4GHz, 25.6MB cache, 192GB memory).
%Only one core was used in the experiments and the operating system was Linux.
%server at the Clemson University School of Computing
%with 
%96~GB memory and a 64-bit Intel Xeon CPU E5645 operating at 2.40GHz. 
%Dell T7500
%workstation with 96~GB memory and Intel Xeon 5600 series processors operating
%at 3.47~GHz.
The sorting phase in the solution of the linear programs in MCA was
carried out by calling $std::sort$, the ($O (n\log n)$) sorting routine
implemented in the C++ standard library.
Comparisons were made with the routine
\smallskip
\begin{center}
METIS{\underscore}ComputeVertexSeparator,
\end{center}
\smallskip
available from METIS 5.1.0 \cite{KarypisKumar98e}.
The following options were employed:
\smallskip
\begin{center}
\begin{tabular}{l}
%METIS{\underscore}CTYPE{\underscore}SHEM
%(heavy edge based matching), \\
METIS{\underscore}IPTYPE{\underscore}NODE
(coarsest problem solved with node growth scheme), \\
METIS{\underscore}RTYPE{\underscore}SEP2SIDED
(Fiduccia-Mattheyses-like refinement scheme).
\end{tabular}
\end{center}
\smallskip

In a preliminary experiment, we
also considered the refinement option METIS{\underscore}RTYPE{\underscore}
SEP1SIDED, but the results obtained were not significantly different. On the average, the option METIS{\underscore}RTYPE{\underscore}SEP2SIDED provides slightly better results.
Additionally, we considered both heavy-edge matching
(METIS{\underscore}CTYPE{\underscore}SHEM) and random matching
(METIS{\underscore}CTYPE{\underscore}RM).

For our experiments, we considered 59 sparse graphs
with dimensions ranging from $n = 1,000$ to $n = 1,965,206$
and sparsities ranging
from $1.43\times10^{-6}$ to $1.32\times10^{-2}$, where
sparsity is defined as the ratio $\frac{|\C{E}|}{n(n-1)}$ (recall that
$|\C{E}|$ is equal to twice the number of edges).
Twenty of these graphs correspond to the adjacency matrix for
symmetric matrices from the
University of Florida Sparse Matrix Collection \cite{UFcollection}.
%For each matrix, a graph was constructed by identifying the columns of the
%matrix with vertices in a graph and drawing edges between vertices $i$ and $j$
%whenever the $(i,j)th$ entry in the matrix was non-zero.
Column 2 of Table
\ref{table1} gives the number of vertices for each of these graphs, followed by
the number of edges ($|\C{E}|/2$),
the sparsity, and the minimum, maximum, and average vertex degrees,
respectively.

Eight large graphs with heavy-tailed degree distribution (HTDD, i.e., there is a large gap between minimum and maximum vertex degree)
were selected from the Stanford SNAP database \cite{snapnets}
(see Table~\ref{table2}).
%found at the following
%web site:
%\smallskip
%\begin{center}
%http://snap.stanford.edu/data/ .
%\end{center}
%\smallskip
%With the exception of p2p-Gnutella06 and Peko01,
%these graphs have heavy-tailed distributions
%(large gap between minimum and maximum vertex degree).

Fifteen graphs (see Table~\ref{table3})
were taken from \cite{Safro12}, and were designed to be
especially challenging for multilevel graph partitioners.
The challenge in these graphs derives
from the fact that the optimal vertex separator is sparse,
yet densely connected to the two shores ($\C{A}$ and $\C{B}$). Also, these graphs represent mixtures of different structures (similar to multi-mode networks) which makes the coarsening uneven.

The Vertex Separator Problem is of particular importance in cyber security.
For example, it can be used to disconnect a largest connected component
in a network to prevent a possible spread of an attack or to find
non-robust structures.
Therefore, we also experimented with a set of $9$ peer-to-peer networks
from SNAP that were collected in \cite{Ripeanu02mappingthe}
(see Table~\ref{table3.5}).
%The number of vertices in the networks ranged from 6301 to 62586
%(see Table~\ref{table3.5})

The UF, HTDD, Hard, and p2p graphs have at most $139,752$ nodes. In order
to assess the performance of BLP on very large scale graphs, we also considered
$7$ graphs from the Konect database \cite{konect} having between $317,080$ and
$1,965,206$ nodes and between $925,872$ and $11,095,298$ edges (see Table
\ref{table3.6}).

\begin{table}
{\small
\begin{center}
\begin{tabular}{l r r r r r r}
&&&&\multicolumn{3}{c}{Degree} \\
\hline
Graph & $|\C{V}|$ & $|\C{E}|/2$ & Sparsity & Min & Max & Ave\\
\hline
bcspwr09&1723&2394&1.61E-03&1&14&2.78\\
bcsstk17&10974&208838&3.47E-03&0&149&38.06\\
c-38&8127&34781&1.05E-03&1&888&8.56\\
c-43&11125&56275&9.09E-04&1&2619&10.12\\
crystm01&4875&50232&4.23E-03&7&26&20.61\\
delaunay\_n13&8192&24547&7.32E-04&3&12&5.99\\
Erdos992&6100&7515&4.04E-04&0&61&2.46\\
fxm3\_6&5026&44500&3.52E-03&3&128&17.71\\
G42&2000&11779&5.89E-03&4&249&11.78\\
jagmesh7&1138&3156&4.88E-03&3&6&5.55\\
lshp3466&3466&10215&1.70E-03&3&6&5.89\\
minnesota&2642&3303&9.47E-04&1&5&2.5\\
nasa4704&4704&50026&4.52E-03&5&41&21.27\\
net25&9520&195840&4.32E-03&2&138&41.14\\
netscience&1589&2742&2.17E-03&0&34&3.45\\
netz4504&1961&2578&1.34E-03&2&8&2.63\\
sherman1&1000&1375&2.75E-03&0&6&2.75\\
sstmodel&3345&9702&1.73E-03&0&17&5.8\\
USpowerGrid&4941&6594&5.40E-04&1&19&2.67\\
yeast&2361&6646&2.39E-03&0&64&5.63\\
\hline
\end{tabular}
\end{center}
\smallskip

\caption{UF Graphs. \label{table1}}
}
\end{table}

%\newgeometry{margin=1.5cm}
%\begin{landscape}

\begin{table}
{\small
\begin{center}
\begin{tabular}{l r r r r r r}
&&&&\multicolumn{3}{c}{Degree} \\
\hline
Graph & $|\C{V}|$ & $|\C{E}|/2$ & Sparsity & Min & Max & Ave\\
\hline
ca-HepPh&7241&202194&7.71E-03&2&982&55.85\\
email-Enron&9660&224896&4.82E-03&2&2532&46.56\\
email-EuAll&16805&76156&5.39E-04&1&3360&9.06\\
oregon2\_010505&5441&19505&1.32E-03&1&1888&7.17\\
soc-Epinions1&22908&389439&1.48E-03&1&3026&34\\
web-NotreDame&56429&235285&1.48E-04&1&6852&8.34\\
web-Stanford&122749&1409561&1.87E-04&1&35053&22.97\\
wiki-Vote&3809&95996&1.32E-02&1&1167&50.4\\
\hline
\end{tabular}
\end{center}
\smallskip

\caption{Heavy-tailed degree distribution graphs from the SNAP database.
\label{table2}}
}
\end{table}

\begin{table}
{\small
\begin{center}
\begin{tabular}{l r r r r r r}
&&&&\multicolumn{3}{c}{Degree} \\
\hline
Graph & $|\C{V}|$ & $|\C{E}|/2$ & Sparsity & Min & Max & Ave\\
\hline
barth5\_1Ksep\_50in\_5Kout&32212&101805&1.96E-04&1&22&6.32\\
bcsstk30\_500sep\_10in\_1Kout&58348&2016578&1.18E-03&0&219&69.12\\
befref\_fxm\_2\_4\_air02&14109&98224&9.87E-04&1&1531&13.92\\
bump2\_e18\_aa01\_model1\_crew1&56438&300801&1.89E-04&1&604&10.66\\
c-30\_data\_data&11023&62184&1.02E-03&1&2109&11.28\\
c-60\_data\_cti\_cs4&85830&241080&6.55E-05&1&2207&5.62\\
data\_and\_seymourl&9167&55866&1.33E-03&1&229&12.19\\
finan512\_scagr7-2c\_rlfddd&139752&552020&5.65E-05&1&669&7.9\\
mod2\_pgp2\_slptsk&101364&389368&7.58E-05&1&1901&7.68\\
model1\_crew1\_cr42\_south31&45101&189976&1.87E-04&1&17663&8.42\\
msc10848\_300sep\_100in\_1Kout&21996&1221028&5.05E-03&1&722&111.02\\
p0291\_seymourl\_iiasa&10498&53868&9.78E-04&1&229&10.26\\
sctap1-2b\_and\_seymourl&40174&140831&1.75E-04&1&1714&7.01\\
south31\_slptsk&39668&189914&2.41E-04&1&17663&9.58\\
vibrobox\_scagr7-2c\_rlfddd&77328&435586&1.46E-04&1&669&11.27\\
\hline
\end{tabular}
\end{center}
\smallskip

\caption{Hard graphs from {\rm \cite{Safro12}}. \label{table3}}
}
\end{table}

\begin{table}
{\small
\begin{center}
\begin{tabular}{l r r r r r r}
&&&&\multicolumn{3}{c}{Degree} \\
\hline
Graph & $|\C{V}|$ & $|\C{E}|/2$ & Sparsity & Min & Max & Ave\\
\hline
p2p-Gnutella04&10879&39994&6.76E-04&0&103&7.35\\
p2p-Gnutella05&8846&31839&8.14E-04&1&88&7.2\\
p2p-Gnutella06&8717&31525&8.30E-04&1&115&7.23\\
p2p-Gnutella08&6301&20777&1.05E-03&1&97&6.59\\
p2p-Gnutella09&8114&26013&7.90E-04&1&102&6.41\\
p2p-Gnutella24&26518&65369&1.86E-04&1&355&4.93\\
p2p-Gnutella25&22687&54705&2.13E-04&1&66&4.82\\
p2p-Gnutella30&36682&88328&1.31E-04&1&55&4.82\\
p2p-Gnutella31&62586&147892&7.55E-05&1&95&4.73\\
\hline
\end{tabular}
\end{center}
\smallskip

\caption{Peer-to-peer networks from {\rm \cite{Ripeanu02mappingthe}}.
\label{table3.5}}
}
\end{table}

\begin{table}
{\small
\begin{center}
\begin{tabular}{l r r r r r r}
&&&&\multicolumn{3}{c}{Degree} \\
\hline
Graph & $|\C{V}|$ & $|\C{E}|/2$ & Sparsity & Min & Max & Ave\\
\hline
out.as-skitter & 1696415 & 11095298 & 7.71E-06 & 1 & 35455 & 13.08\\
out.com-amazon & 334863 & 925872 & 1.65E-05 & 1 & 549 & 5.53\\
out.com-dblp & 317080 & 1049866 & 2.09E-05 & 1 & 343 & 6.62\\
out.com-youtube & 1134890 & 2987624 & 4.64E-06 & 1 & 28754 & 5.27\\
%out.livejournal-links & 5204176 & 0 & 0.00E+00 & 0 & 0 & 0\\
out.roadNet-CA & 1965206 & 2766607 & 1.43E-06 & 1 & 12 & 2.82\\
out.roadNet-PA & 1088092 & 1541898 & 2.60E-06 & 1 & 9 & 2.83\\
out.roadNet-TX & 1379917 & 1921660 & 2.02E-06 & 1 & 12 & 2.79\\
\hline
\end{tabular}
\end{center}
\smallskip

\caption{Konect graphs from {\rm \cite{konect}}.
\label{table3.6}}
}
\end{table}

Vertex costs $c_i$ and weights $w_i$ were assumed to be $1$ at the finest
level for all graphs.
For the bounds on the shores of the separator, we took
$\ell_a = \ell_b = 1$ and $u_a = u_b = \lfloor 0.6 n \rfloor$, which
are the default bounds used by METIS.
Here $\lfloor r \rfloor$ denotes the largest integer not greater than $r$.
%Values for the upper bounds were chosen in order to facilitate
%comparisons with METIS.
In order to enable future comparisons with our solver, we have made this
benchmark set available at http://people.cs.clemson.edu/$\sim$isafro/data.html.

%\end{landscape}
%\restoregeometry

\subsection{Refinement comparison}\label{subsect1}
In this subsection, we compare
%In Proposition \ref{FMprop} we observed that a strongly FM-optimal partition is
%locally optimal in (\ref{CVSP}). This might seem to suggest that an FM-style
an FM-style refinement to a refinement based upon solving (\ref{CVSP})
using the following two experiments:
\begin{itemize}
\item[1.] Obtain an approximate solution
$(\tilde{\m{x}},\tilde{\m{y}})$ to the VSP by calling the multilevel algorithm
METIS{\underscore}ComputeVertexSeparator. Next, refine
$(\tilde{\m{x}},\tilde{\m{y}})$ by calling MCA{\underscore}GR.
\item[2.] Obtain an approximate solution
$(\tilde{\m{x}},\tilde{\m{y}})$ to the VSP by invoking the multilevel
algorithm BLP.
Next, refine $(\tilde{\m{x}},\tilde{\m{y}})$ by calling
METIS{\underscore}NodeRefine.
\end{itemize}
METIS{\underscore}NodeRefine is a refinement routine which
improves upon an initial solution using a Fiduccia-Mattheyses-style
refinement: Vertices in $\C{S}$ are moved into either $\C{A}$ or $\C{B}$
and their neighbors in the opposite shore are moved into $\C{S}$. Vertices
having the largest gains are moved first.
%The iteration ends when every vertex in $\C{S}$ has been moved out. 

The results of the two
experiments are given in Table \ref{table7}. Columns labeled
MCA{\underscore}GR give the average, minimum, and maximum improvement in the size of
$\C{S}$ from calling MCA{\underscore}GR in
Experiment 1, and the last three columns give the results for
Experiment 2.
Here, the improvement is expressed as a percentage using the formula
$100(\C{C}(\C{S}_{{\rm initial}}) - \C{C}(\C{S}_{{\rm final}}))/\C{C}(\C{V})$.
In Experiment 1, we observed that in every initial solution obtained by
METIS{\underscore}ComputeVertexSeparator,
the upper bounds on both of the sets $\C{A}$ and $\C{B}$
were inactive, which we can show implies that the METIS solution
is a local minimizer in the continuous quadratic program (\ref{CVSP}).
%This suggests that the solutions obtained were
%strongly FM-optimal, since vertex costs were equal to 1, and therefore
%all vertices in the separator were eligible to be moved out on the last
%iteration of METIS{\underscore}NodeRefine.
%Therefore, as predicted by Proposition \ref{FMprop}, each initial guess
%was already locally optimal in (\ref{CVSP}).
Hence, the algorithm MCA was unable to improve upon the METIS solutions.
However, MCA{\underscore}GR
%after incorporating
%$\m{c}$-perturbations and $\gamma$-reductions, MCA
gave improvements
of $0.16\%$ on average, compared to only $0.12\%$ in Experiment 2, when
METIS{\underscore}NodeRefine was used.
The greatest improvements achieved by MCA{\underscore}GR
are seen in the HTDD and p2p graphs.
METIS{\underscore}NodeRefine was the most effective on the UF and Hard graphs.
% Refinement Experiments %
\begin{table}
\begin{center}
\begin{tabular}{r | r r r | r r l}
Graph Type & \multicolumn{3}{c|}{MCA{\underscore}GR} &
\multicolumn{3}{c}{METIS{\underscore}NodeRefine} \\
\hline
& avg & min & max & avg & min & max\\
\hline
UF&0.02&0.00&0.22&0.03&0.00&0.38\\
HTDD&0.36&0.00&2.06&0.06&0.00&0.55\\
Hard&0.07&0.00&0.54&0.16&0.00&1.64\\
p2p&0.49&0.03&1.20&0.27&0.04&0.58\\
\hline
Total&0.16&0.00&2.06&0.12&0.00&1.64\\
\hline
\end{tabular}
\end{center}
\smallskip
\caption{Percent improvement in separator sizes using MCA{\underscore}GR
or METIS{\underscore}NodeRefine.\label{table7}}
\end{table}

\subsection{Multilevel solution comparison}\label{subsect2}

Tables \ref{table4}--\ref{tab:konect}
compare the costs $\C{C}(\C{S})$ of the vertex separators found 
by the multilevel implementations BLP and METIS.
The coarsening phases of both algorithms involve matching vertices, which
depends on a random seed.
Therefore
100 trials (with different random seeds) were run for each graph.
The tables report
the average, minimum, and maximum costs obtained by each algorithm.
Columns labeled METIS{\underscore}RM and BLP{\underscore}RM give
the results of using METIS or BLP with random
matching, and columns labeled METIS{\underscore}HE and BLP{\underscore}HE
indicate heavy-edge based matching.
The columns labeled BLP{\underscore}RMFM and BLP{\underscore}HEFM
correspond to a hybrid approach, in which solutions are refined by first performing
Fiduccia-Mattheyses-like (FM) swaps, followed by MCA{\underscore}GR.
FM swaps were performed by calling METIS{\underscore}NodeRefine.
Due to large running times of BLP on the Konect test set, only
BLP{\underscore}RM and
METIS{\underscore}RM were compared for these graphs.

The data in Tables \ref{table4}--\ref{tab:konect} is summarized in Tables
\ref{table_avgs4}--\ref{table_avgs3}. For instance, Table \ref{table_avgs4}
gives the percentage of graphs of each type for which the average size of the
separator obtained by BLP{\underscore}RM was strictly better than METIS{\underscore}RM
(\% Wins)
and the average, minimum, and maximum percentage improvement in the
average separator size compared to METIS, as measured by the expression
$100(\C{C}(\C{S}_{{\rm METIS}}) - \C{C}(\C{S}_{{\rm BLP}}))/\C{C}(\C{V})$.
Note that here,
neither solution is used as an initial guess for the other algorithm, unlike
the experiments of Section \ref{subsect1}.
Table \ref{table_avgs5} compares BLP{\underscore}RMFM with METIS{\underscore}RM,
and Tables \ref{table_avgs1} and \ref{table_avgs3} compare BLP{\underscore}HE
and BLP{\underscore}HEFM with METIS{\underscore}HE.

First, we note that the average improvement was positive for all versions of
BLP on all graph types, except for BLP{\underscore}HE on the UF graphs and
BLP{\underscore}RM on the Konect graphs. 
%Moreover, in all cases in which
%there was no improvement,
%the average separator size was only slightly worse than METIS.
However, even when METIS's separators are smaller, the difference in size is
often not substantial (less than 0.1\%).
%This suggests that
%even though METIS's separators are smaller than BLP{\underscore}HE on 
%average for the UF graphs, the difference in size is often not
%substantial in these cases.
%On the other hand, when BLP{\underscore}HE finds separators which are
%smaller than METIS, they are often significantly smaller
%(greater than 1\%) . 

The BLP algorithms seem to be the most effective on the p2p, HTDD, and Hard graphs.
Based on our observations from Section \ref{subsect1},
the high performance of BLP
on the p2p and HTDD graphs is probably
due to the refinement algorithm, while the
high performance on the Hard graphs may be attributed (at least partially) to
minor differences in the coarsening schemes used by METIS and BLP.
The p2p
graphs performed exceptionally well, giving an improvement over the METIS
solution in 100\% of the trials.
As expected, the hybrid algorithms performed the best on average.

Due to the large dimension of the Konect graphs, combined with their abnormal
degree distributions, the coarsening scheme of BLP produced a large number
of coarse levels for these graphs (over 200 in some cases), since
on many levels only a small number of vertices could be matched. Since the
computational bottleneck of BLP is the refinement phase,
we decided to refrain from refining a given coarse solution until the number
of vertices increased by a factor of $2$ during the uncoarsening process,
in order to improve the running time of BLP. This is one possible explanation
for the relatively poor performance of BLP for these problems.

%%%%%%%%%%%%%%%%%%%%%%%%%%%%%%%%%%%%%%%%%%%%%%%%%%%%%%%%%%%%%%%%%%%%%%%%%%%%%
%BEGIN MULTILEVEL TABLES
%%%%%%%%%%%%%%%%%%%%%%%%%%%%%%%%%%%%%%%%%%%%%%%%%%%%%%%%%%%%%%%%%%%%%%%%%%%%%
\begin{table}
{\small
\begin{center}
\begin{tabular}{l r | r r r | r r r | r r r}
Graph & Best & \multicolumn{3}{c|}{METIS{\underscore}HE} & \multicolumn{3}{c|}{BLP{\underscore}HE}
& \multicolumn{3}{c}{BLP{\underscore}HEFM}\\
\hline
& & avg & min & max & avg & min & max & avg & min & max \\
bcspwr09&6&7.52&6&13&12.20&6&25&9.92&6&22\\
bcsstk17&126&143.88&132&186&158.22&126&234&146.24&126&216\\
c-38&12&14.20&12&22&41.90&14&103&24.16&12&54\\
c-43&83&141.67&103&155&135.76&115&166&108.17&83&138\\
crystm01&65&67.31&65&90&77.64&65&85&73.93&65&85\\
delaunay\_n13&69&74.02&69&83&82.70&72&126&78.61&69&92\\
Erdos992&58&108.07&95&125&72.57&64&82&69.10&58&84\\
fxm3\_6&42&53.48&42&88&66.73&42&90&52.96&42&87\\
G42&412&440.97&424&462&452.28&438&469&441.38&412&466\\
jagmesh7&14&14.03&14&15&19.90&14&42&21.67&14&42\\
lshp3466&51&55.41&51&61&58.96&51&105&52.33&51&73\\
minnesota&14&16.80&14&23&20.75&14&40&19.30&14&35\\
nasa4704&163&176.61&163&206&196.23&168&274&180.25&165&262\\
net25&510&597.34&510&915&557.27&510&993&516.13&510&578\\
netscience&0&0.09&0&3&0.00&0&0&0.00&0&0\\
netz4504&16&18.01&17&20&21.42&16&41&19.93&16&33\\
sherman1&18&29.98&28&39&21.23&18&28&19.88&18&26\\
sstmodel&20&24.28&22&35&27.34&21&37&24.77&20&33\\
USpowerGrid&8&8.98&8&14&20.09&8&44&12.19&8&22\\
yeast&137&192.62&182&213&165.16&156&178&147.23&137&155\\
\hline
Graph & Best & \multicolumn{3}{c|}{METIS{\underscore}RM} & \multicolumn{3}{c|}{BLP{\underscore}RM}
& \multicolumn{3}{c}{BLP{\underscore}RMFM}\\
\hline
%& & avg & min & max & avg & min & max & avg & min & max \\
bcspwr09&6&7.47&6&11&11.88&7&28&9.82&6&18\\
bcsstk17&126&147.29&138&168&153.45&126&356&156.01&126&346\\
c-38&12&24.82&12&72&51.37&14&97&26.32&12&57\\
c-43&83&140.86&117&156&133.44&108&164&106.46&94&118\\
crystm01&65&66.50&65&90&79.60&65&85&77.57&65&80\\
delaunay\_n13&69&75.49&69&90&86.98&71&125&81.56&69&127\\
Erdos992&58&121.23&109&141&73.08&64&86&69.29&59&82\\
fxm3\_6&42&60.82&42&90&73.55&57&101&67.95&42&99\\
G42&412&441.48&427&458&451.01&440&464&443.23&426&463\\
jagmesh7&14&14.14&14&21&21.23&14&39&19.87&14&49\\
lshp3466&51&55.45&52&61&63.36&51&98&55.14&51&86\\
minnesota&14&17.34&14&23&21.03&15&40&18.70&14&30\\
nasa4704&163&175.45&168&188&174.63&168&208&170.94&166&181\\
net25&510&676.32&641&714&546.63&510&990&528.05&510&621\\
netscience&0&0.16&0&5&0.00&0&0&0.00&0&0\\
netz4504&16&18.29&16&23&21.15&16&36&20.00&16&33\\
sherman1&18&30.70&29&50&21.52&18&30&20.56&18&27\\
sstmodel&20&24.85&22&40&25.99&20&37&24.31&20&34\\
USpowerGrid&8&9.13&8&16&18.84&8&35&12.31&8&23\\
yeast&137&212.37&177&236&165.14&153&178&147.33&137&158\\
\hline
\end{tabular}

\end{center}
\smallskip

\caption{Vertex Separator Costs $\C{C}(\C{S})$ for UF graphs.
\label{table4}}
}
\end{table}

%\newgeometry{margin=1.5cm}
%\begin{landscape}

\begin{table}
{\small
\begin{center}
\begin{tabular}{l r | r r r | r r r | r r r}
Graph & Best & \multicolumn{3}{c|}{METIS{\underscore}HE} & \multicolumn{3}{c|}{BLP{\underscore}HE}
& \multicolumn{3}{c}{BLP{\underscore}HEFM}\\
\hline
& & avg & min & max & avg & min & max & avg & min & max \\
ca-HepPh&583&754.23&668&839&657.29&583&706&678.05&591&736\\
email-Enron&426&687.12&604&804&484.94&426&578&481.25&436&583\\
email-EuAll&5&8.99&5&57&9.41&6&18&7.33&6&12\\
oregon2\_010505&37&58.57&48&70&53.30&41&64&43.59&37&59\\
soc-Epinions1&2382&2975.27&2818&3078&2423.42&2382&2465&2529.81&2457&2576\\
web-NotreDame&132&399.97&270&518&431.95&132&543&415.45&134&505\\
web-Stanford&29&133.52&29&575&415.13&95&815&261.23&72&584\\
wiki-Vote&680&704.86&694&731&716.23&698&764&706.03&680&735\\
\hline
Graph & Best & \multicolumn{3}{c|}{METIS{\underscore}RM} & \multicolumn{3}{c|}{BLP{\underscore}RM}
& \multicolumn{3}{c}{BLP{\underscore}RMFM}\\
\hline
ca-HepPh&583&767.56&683&851&674.28&621&720&683.55&625&745\\
email-Enron&426&709.29&650&839&496.20&451&547&487.84&440&574\\
email-EuAll&5&76.04&5&348&11.99&7&35&10.25&6&28\\
oregon2\_010505&37&79.00&66&113&53.60&46&68&43.67&38&51\\
soc-Epinions1&2382&3072.42&2915&3224&2431.98&2399&2475&2529.72&2467&2572\\
web-NotreDame&132&437.77&274&611&462.24&190&567&419.79&136&489\\
web-Stanford&29&143.73&29&484&384.79&134&495&283.44&63&382\\
wiki-Vote&680&708.97&694&737&716.06&696&768&709.51&680&737\\
\hline
\end{tabular}

\end{center}
\smallskip

\caption{Vertex Separator Costs $\C{C}(\C{S})$ for HTDD graphs. \label{table5}}
}
\end{table}

\begin{table}
{\small
\begin{center}
\begin{tabular}{l r | r r r | r r r | r r r}
Graph & Best & \multicolumn{3}{c|}{METIS{\underscore}HE} & \multicolumn{3}{c|}{BLP{\underscore}HE}
& \multicolumn{3}{c}{BLP{\underscore}HEFM}\\
\hline
& & avg & min & max & avg & min & max & avg & min & max\\
vsp\_barth5\_1Ksep\_50in\_5Kout&987&1329.76&1131&1451&1546.57&1309&1767&1353.79&987&1640\\
vsp\_bcsstk30\_500sep\_10in\_1Kout&528&752.65&552&1228&870.56&570&1510&830.00&550&1644\\
vsp\_befref\_fxm\_2\_4\_air02&270&1072.63&989&1142&284.98&270&302&284.87&275&300\\
vsp\_bump2\_e18\_aa01\_model1\_crew1&3849&4306.55&4264&4343&3978.57&3903&4404&4024.77&3849&4411\\
vsp\_c-30\_data\_data&453&510.31&453&594&555.38&475&656&513.32&460&605\\
vsp\_c-60\_data\_cti\_cs4&2222&2600.78&2525&2684&2930.37&2354&3590&2826.26&2222&3589\\
vsp\_data\_and\_seymourl&1030&1253.12&1148&1341&1150.97&1097&1271&1084.40&1045&1258\\
vsp\_finan512\_scagr7-2c\_rlfddd&4605&7438.78&7216&7697&4871.55&4776&5029&4692.32&4608&4798\\
vsp\_mod2\_pgp2\_slptsk&5739&5859.63&5809&5905&7434.46&5767&9711&6391.40&5739&9091\\
vsp\_model1\_crew1\_cr42\_south31&1838&2216.08&2086&2631&2507.86&2424&2576&1971.33&1838&2013\\
vsp\_msc10848\_300sep\_100in\_1Kout&279&648.16&279&929&723.60&343&1421&669.93&387&1209\\
vsp\_p0291\_seymourl\_iiasa&511&536.25&532&542&516.66&511&528&521.60&513&528\\
vsp\_sctap1-2b\_and\_seymourl&3390&4114.48&3831&4373&3925.59&3732&4390&3715.86&3390&4134\\
vsp\_south31\_slptsk&1971&2054.39&1982&2116&2328.52&2242&2567&2031.13&1971&2081\\
vsp\_vibrobox\_scagr7-2c\_rlfddd&2762&3467.55&3362&3613&2856.32&2801&2992&2867.30&2762&3050\\
\hline
Graph & Best & \multicolumn{3}{c|}{METIS{\underscore}RM} & \multicolumn{3}{c|}{BLP{\underscore}RM}
& \multicolumn{3}{c}{BLP{\underscore}RMFM}\\
\hline
vsp\_barth5\_1Ksep\_50in\_5Kout&987&1346.14&1043&1530&1549.01&1326&1818&1361.58&1123&1613\\
vsp\_bcsstk30\_500sep\_10in\_1Kout&528&636.70&528&844&627.95&565&854&629.66&570&850\\
vsp\_befref\_fxm\_2\_4\_air02&270&1464.03&1328&1584&288.35&274&314&285.45&275&303\\
vsp\_bump2\_e18\_aa01\_model1\_crew1&3849&4378.39&4280&4793&3986.52&3897&4389&4070.99&3894&4429\\
vsp\_c-30\_data\_data&453&536.73&471&611&530.51&458&602&500.92&463&568\\
vsp\_c-60\_data\_cti\_cs4&2222&2636.81&2384&2741&2869.78&2345&3545&2666.91&2226&3365\\
vsp\_data\_and\_seymourl&1030&1243.14&1091&1347&1139.91&1087&1207&1082.04&1030&1248\\
vsp\_finan512\_scagr7-2c\_rlfddd&4605&7610.25&7400&7883&4975.58&4861&5170&4691.99&4605&4776\\
vsp\_mod2\_pgp2\_slptsk&5739&5884.52&5838&5925&6123.36&5776&8920&6605.61&5741&9133\\
vsp\_model1\_crew1\_cr42\_south31&1838&2740.18&2558&2894&2508.41&2413&2581&1973.98&1926&2013\\
vsp\_msc10848\_300sep\_100in\_1Kout&279&523.70&279&715&553.37&279&791&541.28&279&785\\
vsp\_p0291\_seymourl\_iiasa&511&535.98&531&545&516.61&511&533&521.36&513&529\\
vsp\_sctap1-2b\_and\_seymourl&3390&4269.77&3884&4557&3921.85&3766&4030&3764.07&3418&4159\\
vsp\_south31\_slptsk&1971&2416.25&2350&2498&2298.76&2205&2563&2033.78&1977&2082\\
vsp\_vibrobox\_scagr7-2c\_rlfddd&2762&4182.86&3995&5240&2878.12&2804&2995&2877.13&2789&2998\\
\hline
\end{tabular}
\end{center}
\smallskip

\caption{Vertex Separator Costs $\C{C}(\C{S})$ for hard graphs.}
\label{table6}
}
\end{table}

\begin{table}
\begin{center}
\begin{tabular}{l r | r r r | r r r | r r r}
Graph & Best & \multicolumn{3}{c|}{METIS{\underscore}HE} & \multicolumn{3}{c|}{BLP{\underscore}HE}
& \multicolumn{3}{c}{BLP{\underscore}HEFM}\\
\hline
& & avg & min & max & avg & min & max & avg & min & max \\
p2p-Gnutella04&1656&2055.11&1986&2103&1710.59&1675&1755&1697.08&1662&1738\\
p2p-Gnutella05&1306&1666.38&1629&1724&1369.70&1340&1404&1348.47&1306&1385\\
p2p-Gnutella06&1253&1605.98&1568&1653&1305.96&1265&1343&1291.29&1260&1327\\
p2p-Gnutella08&771&1009.10&976&1043&825.25&794&855&795.99&772&821\\
p2p-Gnutella09&975&1287.01&1253&1327&1044.83&1021&1081&1003.17&975&1038\\
p2p-Gnutella24&2463&3284.91&3203&3380&2728.58&2671&2781&2511.26&2467&2553\\
p2p-Gnutella25&2043&2761.52&2691&2836&2279.78&2227&2345&2089.72&2043&2143\\
p2p-Gnutella30&3016&4267.82&4094&4398&3320.50&3245&3422&3097.88&3035&3176\\
p2p-Gnutella31&4905&5985.32&5888&6184&5581.26&5460&5750&5002.65&4905&5081\\
\hline
Graph & Best & \multicolumn{3}{c|}{METIS{\underscore}RM} & \multicolumn{3}{c|}{BLP{\underscore}RM}
& \multicolumn{3}{c}{BLP{\underscore}RMFM}\\
\hline
p2p-Gnutella04&1656&2140.50&2089&2200&1708.87&1673&1749&1696.39&1656&1742\\
p2p-Gnutella05&1306&1720.54&1687&1750&1369.40&1341&1407&1350.60&1324&1389\\
p2p-Gnutella06&1253&1689.01&1641&1727&1306.65&1275&1339&1290.61&1253&1318\\
p2p-Gnutella08&771&1042.30&1003&1075&825.46&795&857&794.54&771&822\\
p2p-Gnutella09&975&1328.33&1293&1367&1044.11&1010&1076&1003.04&981&1038\\
p2p-Gnutella24&2463&3617.52&3538&3685&2727.85&2670&2806&2514.59&2463&2563\\
p2p-Gnutella25&2043&3008.89&2931&3095&2276.96&2233&2335&2091.69&2046&2144\\
p2p-Gnutella30&3016&4692.64&4580&4798&3321.06&3239&3433&3093.14&3016&3184\\
p2p-Gnutella31&4905&6904.14&6619&7468&5571.38&5432&5712&5014.85&4929&5071\\
\hline
\end{tabular}
\end{center}
\caption{Vertex Separator Costs $\C{C}(\C{S})$ for peer-to-peer networks.
\label{tab:p2p}}
\end{table}

\begin{table}
\begin{center}
\begin{tabular}{l r | r r r | r r r }
Graph & Best & \multicolumn{3}{c|}{METIS{\underscore}RM} & \multicolumn{3}{c}{BLP{\underscore}RM}\\
\hline
& & avg & min & max & avg & min & max \\
out.as-skitter&15006&20654.00&15006&23149&30834.29&26728&32380\\
out.com-amazon&4237&4470.61&4237&4644&5622.31&5386&5893\\
out.com-dblp&10600&10859.04&10600&11136&12762.98&11855&13848\\
out.com-youtube&16729&28355.39&26911&30371&17480.22&16729&18165\\
%out.livejournal-links&0&0.00&0&0&0.00&0&0\\
out.roadNet-CA&99&124.13&99&168&336.60&127&1222\\
out.roadNet-PA&109&128.90&109&155&343.11&131&865\\
out.roadNet-TX&59&77.41&59&130&248.71&74&533\\
\hline
\end{tabular}
\end{center}
\caption{Vertex Separator Costs $\C{C}(\C{S})$ for Konect networks.
\label{tab:konect}}
\end{table}

\begin{table}
\begin{center}
\begin{tabular}{l | r | r r r}
& &\multicolumn{3}{c}{\% Improvement}\\
\hline
Graph Type & \% BLP Wins & avg & min & max \\ 
\hline
UF&35.00&0.10&-0.62&2.00\\
p2p&100.00&3.52&2.13&4.39\\
HTDD&62.50&0.84&-0.20&2.80\\
Hard&73.33&0.96&-0.63&8.33\\
Konect&14.29&-0.09&-0.60&0.96\\
Total&55.93&0.92&-0.63&8.33\\
\hline
\end{tabular}
\end{center}
\caption{Comparison of separator costs $\C{C}(\C{S})$ obtained by BLP{\underscore}RM
and METIS{\underscore}RM. \label{table_avgs4}}
\end{table}

\begin{table}
\begin{center}
\begin{tabular}{l | r | r r r}
& &\multicolumn{3}{c}{\% Improvement}\\
\hline
Graph Type & \% BLP Wins & avg & min & max \\ 
\hline
UF&45.00&0.26&-0.50&2.75\\
p2p&100.00&4.04&3.02&4.57\\
HTDD&74.22&1.22&-0.11&3.36\\
Hard&77.81&1.33&-0.08&8.35\\
Total&68.48&1.37&-0.50&8.35\\
\hline
\end{tabular}
\end{center}
\caption{Comparison of separator costs $\C{C}(\C{S})$ obtained by BLP{\underscore}RMFM
and METIS{\underscore}RM.
\label{table_avgs5}}
\end{table}

\begin{table}
\begin{center}
\begin{tabular}{l | r | r r r}
& &\multicolumn{3}{c}{\% Improvement}\\
\hline
Graph Type & \% BLP Wins & avg & min & max \\ 
\hline
UF&30.00&-0.02&-0.57&1.16\\
p2p&100.00&2.59&0.65&3.44\\
HTDD&50.00&0.67&-0.30&2.41\\
Hard&46.67&0.38&-1.55&5.58\\
Total&50.00&0.65&-1.55&5.58\\
\hline
\end{tabular}
\end{center}
\caption{Comparison in separator costs $\C{C}(\C{S})$ obtained by BLP{\underscore}HE
and METIS{\underscore}HE. \label{table_avgs1}}
\end{table}

\begin{table}
\begin{center}
\begin{tabular}{l | r | r r r}
& &\multicolumn{3}{c}{\% Improvement}\\
\hline
Graph Type & \% BLP Wins & avg & min & max \\ 
\hline
UF&40.00&0.17&-0.67&1.92\\
p2p&100.00&3.11&1.57&3.61\\
HTDD&62.50&0.66&-0.10&2.13\\
Hard&60.00&0.75&-0.52&5.58\\
Total&59.62&0.92&-0.67&5.58\\
\hline
\end{tabular}
\end{center}
\caption{Comparison of separator costs $\C{C}(\C{S})$ obtained by BLP{\underscore}HEFM
and METIS{\underscore}HE.
\label{table_avgs3}}
\end{table}

%\end{landscape}
%\restoregeometry

In some applications, such as cybersecurity, the size of a vertex separator
is of primary importance, while in other applications, such as sparse matrix
reordering, small separators must be found quickly. 
Table \ref{table_times}
compares the average, geometric mean, minimum, and maximum CPU time (in seconds)
for BLP{\underscore}RM and METIS{\underscore}RM
on each of the five test sets.
For the first four test sets in this table, the average CPU time for BLP was
on average about 260 times slower than METIS.
However, we note that the average BLP solution was strictly better than the
best METIS solution (over 100 trials) in 25 out the 52 instances in these
test sets, and in fact for all 9 p2p instances.
For the Konect graphs, the CPU time gap between BLP and METIS was approximately
28 times. We suspect that this relative improvement in CPU time is due to the
reduced number of refinement phases used for these problems. 
%In the Hard graph set, the running time
%of BLP was the longest on vsp\_mod2\_pgp2\_slptska (approximately $1110$
%seconds). If we omit the running time for this problem, the average CPU time
%for BLP on the Hard graphs is 31.93 (s) and the Total average CPU time
%is 15.22, while the average CPU time for METIS becomes 0.12.
%Thus, the current implementation of BLP is between 125 and 150 times
%slower than METIS on average.
%However, we note that the average BLP solution was better than the best
%METIS solution (over 100 trials) in 27 out of 52 instances, and in fact for
%all 9 p2p instances.
% comparing the best METIS solution (over 100 trials) to the average
%BLP solution, the BLP solution is better in 22 out of 52 instances.
%\emph{However, adding more runs of METIS in order to make the total running time
%of BLP and METIS the same does not improve the minimum and
%average results of METIS.}
%We also note that experimenting with the same implementation
%on a different hardware configuration, namely, 96~GB memory on
%a 64-bit Intel Xeon CPU E5645 operating at 2.40GHz,
%reduces the difference to approximately 100 times.
%The table indicates that the current version
%of our algorithm is considerably slower than METIS.
Finally, we stress that BLP has not been optimized for speed.
Figure~\ref{graph} gives a
log-log plot of $n = |\C{V}|$ versus CPU time for all 59 instances.
The best fit line through the data in the log-log plot has a slope of
approximately 1.65, which indicates that the CPU time of BLP is
between a linear and a quadratic function of the number of vertices.

\begin{table}
\begin{center}
\begin{tabular}{l | r r r r | r r r r}
& \multicolumn{4}{c|}{BLP{\underscore}RM} & \multicolumn{4}{c}{METIS{\underscore}RM}\\
\hline
Graph Type & avg & geomean & min & max & avg & geomean & min & max\\ 
\hline
UF&0.56&0.34&0.03&3.08&0.01&0.00&0.00&0.12\\
p2p&36.08&14.99&1.48&276.09&0.16&0.13&0.05&0.49\\
HTDD&18.34&7.64&0.64&104.82&0.13&0.08&0.01&0.55\\
Hard&88.59&32.23&1.58&719.48&0.28&0.21&0.05&0.84\\
Konect&9258.19&5989.74&689.58&27888.17&334.55&4.97&0.63&2702.44\\
Total&1264.62&10.44&0.03&27888.17&44.72&0.00&0.00&2702.44\\
\hline
\end{tabular}
\end{center}
\caption{CPU times (in seconds) for each algorithm on different graph types.
\label{table_times}}
\end{table}

%%%%%%%%%%%%%%%%%%%%%%%%%%%%%%%%%%%%%%%%%%%%%%%%%%%%%%%%%%%%%%%%%%%%%%%%%%%%%
%END MULTILEVEL TABLES
%%%%%%%%%%%%%%%%%%%%%%%%%%%%%%%%%%%%%%%%%%%%%%%%%%%%%%%%%%%%%%%%%%%%%%%%%%%%%

In order to determine the computational bottlenecks of BLP, we
examined a flat profile of the code, using the Linux utility GNU gprof.
We randomly selected one problem from each of the first
four test sets and found that
between 61\% and 87\% of the CPU time was consumed by the routine which
implements the greedy algorithm for solving the linear programs in MCA.
The remainder of the CPU time was shared by objective value computations,
matrix vector product computations (between $\m{A}$ and $\m{x}$ or $\m{y}$),
and Karush-Kuhn-Tucker multiplier computations, which were used to 
determine the perturbations in
$\m{c}$ or $\gamma$ required to escape a local optimum.

Thus, for applications in which speed is more important than solution quality,
the following modifications may be investigated:
\begin{itemize}
\item [1.] Instead of resolving the linear program in MCA from scratch,
we could exploit the structure of the previously computed solution to update it.
\item [2.] In each iteration of the Mountain Climbing Algorithm, we need
the products $\m{Ax}_k$ and $\m{Ay}_k$ between the matrix and a vector.
We could save the previous products $\m{Ax}_{k-1}$ and $\m{Ay}_{k-1}$ and
only recompute the parts of the products that change.
%\item [3.] The $\m{c}$-perturbations could be augmented by an analysis of
%the second-order sufficient optimality conditions given in
%\cite{HagerHungerford14}.
\end{itemize}

\begin{figure}[t]
\begin{center}
\includegraphics[scale=0.5]{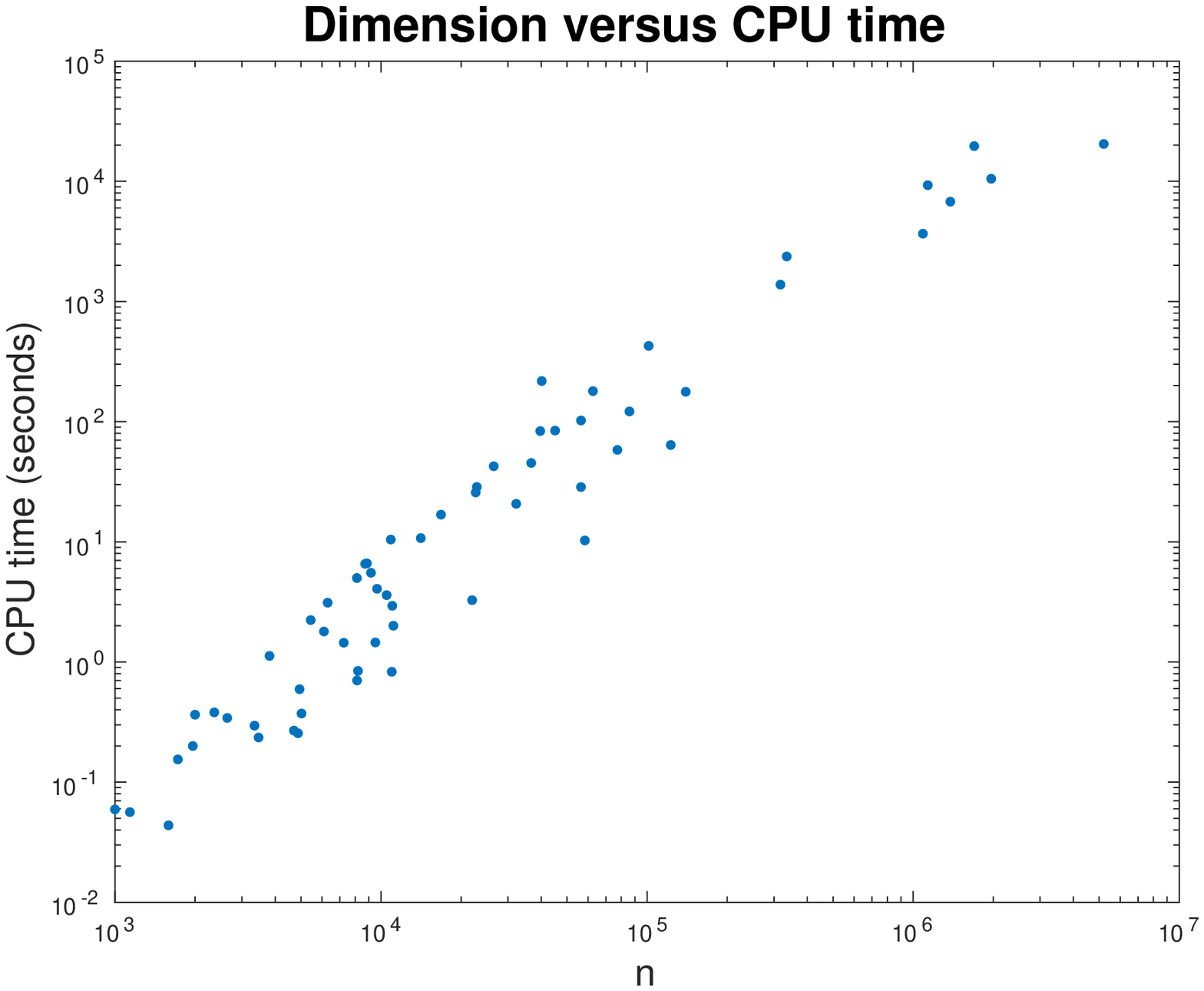}
\end{center}
\caption{Number of vertices $n$ versus CPU time for BLP{\underscore}RM.\label{graph}}
\end{figure}

%-------------------------------------------------------------------------------
\section{Conclusion}
\label{sectConclude}
%-------------------------------------------------------------------------------

We have developed a new algorithm (BLP) for solving large-scale instances of the
Vertex Separator Problem (\ref{VSP}).
The algorithm incorporates a multilevel framework; that is,
the original graph is coarsened several times; the problem is solved for the
coarsest graph; and the solution to the coarse graph is gradually uncoarsened
and refined to obtain a solution to the original graph. A key feature of the
algorithm is the use of the continuous bilinear program (\ref{CVSP}) in both the
solution and refinement phases. In the case where vertex weights are all equal
to one (or a constant), (\ref{CVSP}) is an exact formulation of the VSP in the
sense that there exists a binary solution satisfying (\ref{sepcond}),
from which an optimal solution to the VSP can be recovered using (\ref{ABS}).
When vertex weights are not all equal, we showed that (\ref{CVSP}) still
approximates the VSP in the sense that there exists a mostly binary solution.

%The relationship between local optimality in a Fiduccia-Mattheyses sense
%and local optimality in (\ref{CVSP}) was also analyzed. In particular,
%we showed
%that every partition which is strongly FM-optimal is locally optimal in the
%BLP, and that a weak converse also holds.

During the solution and refinement phases of BLP, the bilinear program
is solved approximately by applying the algorithm MCA{\underscore}GR,
a mountain climbing algorithm which incorporates perturbation
techniques to escape from stationary points and explore a new part
of the search space.
One technique, referred to as $\m{c}$-perturbations, uses the
first-order optimality conditions to derive a tiny perturbation that
will force an iterate to a new location with a possibly improved separator.
The second technique, referred to as $\gamma$-perturbations,
improves the separator by relaxing the requirement that
there are no edges between the sets in the partition.
We determined the smallest relaxation that will generate a new partition.
%The final iterate of MCA{\underscore}GR
%is typically a stationary point of (\ref{CVSP}).
%Local optimality
%of this point may be verified by checking the conditions (V2)--(V5); and whenever
%one of these conditions fails, a feasible ascent direction is given by Table
%\ref{ascent}.
To our knowledge, this is the first multilevel algorithm
to make use of a continuous optimization based
refinement method for the family of graph partitioning problems.
The numerical results of Section~\ref{sectResults}
indicate that BLP is capable of locating vertex separators of high quality
(comparing against METIS), and is particularly effective on
p2p graphs, HTDD graphs (graphs with heavy-tailed distributions), and
graphs having relatively sparse separators.

\newpage
\bibliographystyle{siam}

\end{document}